\documentclass{article}
\usepackage{graphicx} 
\usepackage{amsmath}
\usepackage{amssymb}
\usepackage{amsthm}
\usepackage{mathrsfs}
\usepackage{CJKutf8}
\usepackage[utf8]{inputenc}
\usepackage{diagbox}
\usepackage[figuresright]{rotating} 
\usepackage[section]{placeins}
\usepackage{makecell}
\usepackage{longtable}
\allowdisplaybreaks[4]
\usepackage{xcolor}
\usepackage{embrac}
\usepackage{xpatch}
\usepackage{ytableau}
\usepackage{tikz}
\usetikzlibrary{positioning}

\newtheorem{theorem}{Theorem}[section]

\newtheorem{lemma}[theorem]{Lemma}
\newtheorem{proposition}[theorem]{Proposition}
\newtheorem{corollary}[theorem]{Corollary}

\newtheorem{remark}{Remark}[section]
\newtheorem{example}{Example}[section]
\newtheorem{conjecture}{Conjecture}[section]

\def\P{\partial}
\def\bt{\mathbf{t}}
\def\bs{\mathbf{s}}
\def\ba{\mathbf{a}}
\def\bp{\mathbf{p}}

\def\mH{\mathcal{H}}
\def\mZ{\mathcal{Z}}
\def\bx{\mathbf{x}}

\title{Combinatorics and asymptotic behavior for double Hurwitz numbers}
\author{Xiang Li
	\\\small\textsuperscript{1} School of Mathematical Sciences, University of Science and Technology of China,\\\small
	\quad\ \ Hefei 230026, P.R. China
	\\\small Email: lxiang1993@ustc.edu.cn.}
\begin{document}
	\renewcommand{\today}{}
	\maketitle
	\begin{CJK}{UTF8}{gbsn}
		\begin{abstract}
			Polynomial-in-time algorithms for computing classical Hurwitz numbers were given in \cite{DYZ} based on the Pandharipande equation. 
			The paritition function of double Hurwitz numbers was proved \cite{Ok1} to satisfy the 2-Toda hierarchy.
			In this paper, similar to \cite{Ok1} we derive Pandharipande-type equations for double Hurwitz numbers from 2-Toda hierarchy. Based 
			on these equations and a method from \cite{DYZ}, we study large genus as well as large degree asymptotics of double Hurwitz numbers.
		\end{abstract}
		\section{Introduction}
		The notion of Hurwitz numbers was introduced in \cite{Hur1}, \cite{Hur2}. The question is to count the weighted number $H_d^*(\mu^{(1)}, \cdots , \mu^{(m)})$ of ramified coverings of degree $d$ of $\mathbb{P}^1$ with the ramification profiles $\mu^{(1)},\cdots,\mu^{(m)}\vdash d$. Here, $\mu\vdash d$ denotes a partition $\mu$ of weight $d$. The ramified covering is called connected if the upper Riemann surface is connected. Denote by $H_{g,d}(\mu^{(1)}, \cdots , \mu^{(m)})$ the weighted number of connected ramified covering of genus $g$ and degree $d$ with the ramification profiles $\mu^{(1)},\cdots,\mu^{(m)}\vdash d$. For a partition $\mu=(\mu_1,\,\cdots,\,\mu_{l})\vdash d$ with $\mu_1\geq\mu_2\geq\cdots\geq\mu_{l}> 0$, denote $|\mu|=\sum_{i=1}^{l}\mu_{i}=d,\,l(\mu)=l$ and $l^*(\mu)=|\mu|-l(\mu)$. If all $\mu_i\geq2$, we write $\mu\models d$. For two partitions $\mu$ and $\nu$, their union $\mu\cup \nu$ is defined as the partition obtained by combining the parts of $\mu,\nu$ and arranging them in non-increasing order. We write $\nu\subset\mu$, if there exists a partition $\lambda$, possibly empty, such that $\mu=\nu\cup\lambda$, in which case we denote $\lambda=\mu\setminus\nu$.
		
		It is well known (cf.~\cite{Burn,Fro}) that
		\begin{align}\label{Hurwitz}
			H_d^*(\mu^{(1)}, \cdots , \mu^{(m)})=\sum_{\lambda\vdash d}\left(\frac{\text{dim}\lambda}{d!}\right)^{2} \prod_{i=1}^m f_{\mu^{(i)}}(\lambda),
		\end{align}
		where $\text{dim}\lambda$ is the dimension of the irreducible representations of the symmetric group $S(d)$ corresponding to $\lambda$, and
		\begin{align}
			f_{\mu^{(i)}}(\lambda)=\frac{d!}{z_{\mu^{(i)}}}\frac{\chi^\lambda(\mu^{(i)})}{\text{dim}\lambda}.\label{f}
		\end{align}
		Here $z_\mu=\prod\limits_i m_i(\mu)! i^{m_i}$, with $m_i(\mu)$ being the multiplicity of $i$ in $\mu$, and $\chi^\lambda(\mu^{(i)})$ (cf.~\cite{Mac}) is the value of the irreducible character $\chi^\lambda$ on the conjugacy class $\mu^{(i)}$.
		
		As customary in the literature, we call 
		\begin{align}
			&H^*_{d}(\mu^{(1)},\mu^{(2)},\,\underbrace{2\,1^{d-2},2\,1^{d-2},2\,1^{d-2},\cdots}_{k})=:H_{k,d}^*(\mu^{(1)},\mu^{(2)})
		\end{align}
		\textit{not-necessarily connected double Hurwitz numbers}, and we call $H_{g,d}(\mu^{(1)},\mu^{(2)},\allowbreak 2\,1^{d-2},2\,1^{d-2},\cdots)$ \textit{connected double Hurwitz numbers} (or simply \textit{double Hurwitz numbers}). For notational convenience, we denote
		\begin{align}
			h_{g,d}(\mu^{(1)},\mu^{(2)},\,2\,1^{d-2},2\,1^{d-2},\cdots)=&\frac{H_{g,d}(\mu^{(1)},\mu^{(2)},2\,1^{d-2},2\,1^{d-2},\cdots)}{(2g+2d-l^*(\mu^{(1)})-l^*(\mu^{(2)})-2)!}. 
		\end{align}
		
		Denote $h_{g,d}=h_{g,d}(1^d,1^d,\,2\,1^{d-2},2\,1^{d-2},\cdots)$. In \cite{P}, based on the Toda conjecture, Pandharipande deduced the following equation:
		\begin{align}\label{Pand}
			D^2\mH(x,y)=y e^{\mH(x,y e^x)-2\mH(x,y)+\mH(x,ye^{-x})},\quad D:=y\frac{d\,\,}{dy},
		\end{align}
		where
		$$
		\mH(x,y):=\sum_{g,d}h_{g,d}\, y^d x^{2g+2d-2}.
		$$
		Like in \cite{DYZ}, we call \eqref{Pand} the \textit{Pandharipande equation}. In \cite{DYZ}, 
		by applying $D=y\frac{d~}{dy}$ on both sides equation \eqref{Pand} is simplified to 
		\begin{align}
			D^3\mH(x,y)-&D^2\mH(x,y)\nonumber  \\
			&\,=D^2\mH(x,y)(D\mH(x,e^xy)-2D\mH(x,y)+D\mH(x,e^{-x}y)).\label{DYZf2}
		\end{align}
		
		More generally, define
		\begin{align}
			\mH(\bp,\bp';x,y):=\sum_{g,d}\sum_{\mu^{(1)},\,\mu^{(2)}\vdash d}&y^dh_{g,d}(\mu^{(1)},\mu^{(2)},\,2\,1^{d-2},2\,1^{d-2},\cdots) \nonumber\\ 
			& \qquad\times   p_{\mu^{(1)}}  p'_{\mu^{(2)}}x^{2g+2d-l^*(\mu^{(1)})-l^*(\mu^{(2)})-2},\label{Hxy}
		\end{align}
		and let
		\begin{align}
			Z(\bp,\bp';x,y)=\exp\left(\mH(\bp,\bp';x,y)\right)\label{rlt}
		\end{align}
		be the partition function of double Hurwitz numbers, where $p_\lambda=\prod_{i=1}^{l(\lambda)}p_{\lambda_i}$. It is proved by Okounkov \cite{Ok1} that 
		\begin{align}
			\tau_n(\bt_+,\bt_-;x,y)=y^{\frac{n^2}{2}}e^{\frac{n(4n^2-1)x}{24}}Z(\bp,\bp';x,y)\label{tau_r}
		\end{align}
		is a 2-Toda $\tau$-function, namely, $\tau_n(\bt_+,\bt_-;x,y)$ satisfies the bilinear identity \cite{UT}
		\begin{align}
			&\oint_{C_\infty}\frac{dz}{2\pi i}z^{-(n-n')}e^{\xi(\bt_+-\bt'_+,z)}\tau_n(\bt_+-[z^{-1}],\bt_-;x,y) \nonumber\\
			&\qquad\qquad\qquad\qquad\qquad\qquad\qquad\qquad\qquad\times\tau_{n'}(\bt'_++[z^{-1}],\bt'_-;x,y)\nonumber\\
			&\quad=\oint_{C_0}\frac{dz}{2\pi i}z^{n-n'}e^{-\xi(\bt_--\bt'_-,z)}\tau_{n+1}(\bt_+,\bt_-+[z^{-1}];x,y)\nonumber\\
			&\qquad\qquad\qquad\qquad\qquad\qquad\qquad\qquad\qquad\times\tau_{n'-1}(\bt'_+,\bt'_--[z^{-1}];x,y),\label{Hir}
		\end{align}
		where $[z]=(z,\,\frac{z^2}{2},\,\frac{z^3}{3},\,\cdots)$, $p_i=it_{i},\,p'_i=it_{-i}$, $C_0$ denotes a small counterclockwise contour around $0$ and $C_\infty$ denotes a clockwise contour around $\infty$. Using the first non-trivial equation contained in \eqref{Hir}, Okounkov \cite{Ok1} proved the validity of the Pandharipande equation \eqref{Pand}. 
		
		For fixed $U\models t,\,V\models w,\,t,w\in \mathbb{N}$, we denote
		\begin{align}
			\mH(U,V;x,y):=&\sum\limits_{g,d} h_{g,d}(U\,1^{d-t},\,V\,1^{d-w},\,2\,1^{d-2},\,2\,1^{d-2},\cdots) \nonumber\\
			&\qquad\qquad\qquad\qquad\qquad\times y^d x^{2g+2d-l^*(U)-l^*(V)-2}, \label{Huvxy}\\
			\mH^{[k,l]}(U,V;x,y):=&\frac{y^{t+k}}{k!l!}\frac{d^{k}\,\,}{dy^{k}}\left(\frac{y^{w+l}}{y^{t}}\frac{d^{l}\,\,}{dy^{l}}\left(\frac{1}{y^{w}}\mH(U,V;x,y)\right)\right).\label{DHuvxy}
		\end{align}
		Clearly, $\mH(x,y)=\mH(\varnothing,\varnothing;x,y)$ and $D\mH(x,y)=\mH^{[1,0]}(\varnothing,\varnothing;x,y)$. We take a series of derivatives with respect to $p_2,\,p_3,\,\cdots,\,p'_2,\,p'_3,\,\cdots$ on certain simplifications of other equations contained in \eqref{Hir}, and using an idea similar to Okounkov \cite{Ok1} we put $p_1=p'_1=1,\,p_2=p'_2=p_3=p'_3=\cdots=0$ in the resulting equations. Then we get an equation for $\mH(U,\varnothing;x,y)$ with $U=(u_1,\cdots)\models t$:
		\begin{align}		 
			&\mH^{[0,1]}(U,\varnothing;x,y)=  \sum_{j=1}^{u_1-1} \sum_{\substack{(P,R)\in\operatorname{Part}^j[U]}}\frac{1}{u_1j!m_{u_1}(U)}\mH^{[1+m_1(P^{(0)}),1]}( P^{(0)*},\varnothing;x,y)   \nonumber   \\ 
			&\times  \prod_{i=1}^{j} \mathcal{C}^{P^{(i)}}_{R^{(i)}}\Big(\mH^{[m_1(P^{(i)}),0]}(P^{(i)*},\varnothing;x,e^x y)+(-1)^{l(R^{(i)})}\mH^{[m_1(P^{(i)}),0]}(P^{(i)*},\varnothing;x,e^{-x}y) \Big).    \label{Pand4}
		\end{align}
		Here $\operatorname{Part}^j[U]$ denotes the set of tuples\footnote{We thank Zejun Zhou for his help in simplifying the notations here.} $(P,R)$, where $P=\{P^{(0)},\dots,P^{(j)}\}$ and $R=\{R^{(1)},\dots,R^{(j)}\}$ are ordered sets satisfying $P^{(i)}\supset R^{(i)}\neq\emptyset,\,u_1-1=\sum_{i=1}^j\allowbreak|R^{(i)}|$ and $P^{(0)}\cup\bigcup_{i=1}^{j} \big(P^{(i)}\setminus R^{(i)}\big)=U\setminus (u_1)$.
		Let $U^*$ are the parts of $U$ larger than or equal to $2$ and define $\mathcal{C}^U_V=\prod_{k}\binom{m_k(U)}{m_k(V)}$.
		
		We also get an equation for $\mH(U,V;x,y)$ with $U=:(u_1,\cdots)\models t,\,V=:(v_1,\cdots)\models w,\,t,w\in\mathbb{N}$: 
		\begin{align}      
			&\mH(U,V;x,y)=\sum_{r=1}^{v_1-1}\sum_{i_0=1}^{r+1}\sum_{\substack{P\in\operatorname{Split}_{i_0}^{r+\delta^{i_0}_{r+1}}[U]\\(\overline{P},S)\in\operatorname{Part}^{r+\delta^{i_0}_{r+1}}[V]   }}               \frac{(u_1-1) m_{u_1-1}(P^{(i_0)})}{2v_1 r!m_{u_1}(U)m_{v_1}(V)}  \nonumber    \\	
			& \times\mH^{[1+m_1({P^{(0)}}),1+m_1(\overline{P}^{(0)})]}(P^{(0)*},\overline{P}^{(0)*};x,y)\prod_{i=1}^{r+\delta^{i_0}_{r+1}}\mathcal{C}^{\overline{P}^{(i)} }_{S^{(i)} }(\mH^{[m_1({P^{(i)}}),m_1({\overline{P}^{(i)}})]}( P^{(i)*} ,\overline{P}^{(i)*} ;x,e^x y)    \nonumber    \\			
			&  +(-1)^{l(S^{(i)} )+\delta^i_{i_0}}\mH^{[m_1({P^{(i)}}),m_1({\overline{P}^{(i)}})]}( P^{(i)*} ,\overline{P}^{(i)*}  ;x,e^{-x}y))-\sum_{w=1}^{[\frac{u_1}{2}]+1}\sum_{j_0=1}^{w}\sum_{\substack{(P,R)\in\overline{\operatorname{Part}}^{w}[U]\\\overline{P}\in\overline{\operatorname{Split}}_{j_0}^{w}[V]}} \nonumber\\			
			& \frac{2^{w-1} m_{v_1}(\overline{P}^{(j_0)}) }{w!m_{u_1}(U)m_{v_1}(V)} \prod_{j=1}^{w}\mathcal{C}^{P^{(j)}}_{R^{(j)}}\mH^{[m_1(P^{(j)}) ,m_1(\overline{P}^{(j)})]}(P^{(j)*} ,\overline{P}^{(j)*} ;x,y).  \label{D_Pand} 
		\end{align}	
		Here $\operatorname{Split}_{k}^j[U]$ denotes the set of ordered sets $P=\{P^{(0)},\dots,P^{(j)}\}$ satisfying $P^{(k)}\supset (u_1-1)$ and $\big(\bigcup_{i=0}^{j} P^{(i)}\big)\setminus (u_1-1)=U\setminus (u_1)$. Similarly, $\overline{\operatorname{Split}}_{k}^j[U]$ denotes the set of ordered sets $P=\{P^{(0)},\dots,P^{(j)}\}$ satisfying $P^{(k)}\supset (u_1)$ and $\big(\bigcup_{i=0}^{j} P^{(i)}\big)=U$. Moreover, $\overline{\operatorname{Part}}^{j}[U]$ is defined analogously to $\operatorname{Part}^{j}[U]$ with the following modifications: $u_1=\sum_{i=1}^j\allowbreak|R^{(i)}|$, and additional constraint that $R^{(j_0)}\neq(u_1),l(R^{(j)})$ is even, while $l(R^{(j_0)})$ is odd. 
		In each step of recursion, $m_{u_1}(U)$ decreases by $1$. Thus, after a finite number of steps, we obtain either $\mH(U,\varnothing;x,y)$ or $\mH(U,V;x,y)$ being the equation of $\mH(x,y)$ which has already been calculated in \cite{DYZ}. Before proceeding, we define for $g\geq0$
		\begin{align}
			\mH_{g}(U,V;z):=&\sum\limits_{d} h_{g,d}(U\,1^{d-t},\,V\,1^{d-w},\,2\,1^{d-2},2\,1^{d-2},\cdots)z^d,    \label{Hg2} 
		\end{align}		
		for $U\models t,\,V\models w,\,t,w\in \mathbb{N}$. Using representation theory, Zvonkine \cite{Zvon} deduced the following
		
		\vspace{3pt}
		\noindent\bf{Theorem~A}\mdseries{} (Zvonkine \cite{Zvon}). \textit{The generating series of double Hurwitz numbers with a fixed genus $H_{g}(U,V ;T)=H_{g}(U,V ;T(z)):=\mH_{g}(U,V ;z)$ belong to $\mathbb{Q}[T,\frac{1}{1-T}]$, where
			\begin{align}
				z & =T(z) e^{-T(z)},\qquad T(z) =\sum_{d=1}^{\infty}\frac{d^{d-1}}{d!}z^d,  \label{z_T}
			\end{align} 
			except for $H_{1}(\varnothing,\varnothing;T)=-\frac{T+\log(1-T)}{24}$.}
		\vspace{3pt}
		\begin{remark}\label{LargeS}
			Zvonkine \cite{Zvon} gave more general results. Similar results for special cases can be found in \cite{DYZ,GJ4,GJ3,GJV,V}.
		\end{remark}		
		The case with $U=V=\varnothing$ was carefully studied in \cite{DYZ}. More precisely, Dubrovin-Yang-Zagier \cite{DYZ} deduce from \eqref{DYZf2} that
		\begin{align}\label{DYZf1}
			\mH_{g}^{(3)}(z)-\mH_{g}^{(2)}(z)=\sum_{\substack{g_1,g_2\geq0,l\geq 1\\g_1+g_2+l=g+1}}&\frac{2}{(2l)!}\mH_{g_1}^{(2)}(z)\mH_{g_2}^{(2l+1)}(z),\qquad g\geq0,
		\end{align}
		where $z=x^2y$ and
		\begin{align}
			\mH_{g}^{(\ell)}(z)=D^\ell\mH_{g}(z)=\sum_{d}  d^\ell h_{g,d} z^d, \label{DH}
		\end{align}
		with $D=z\frac{d\,\,}{dz}=\frac{T}{1-T}\frac{d\,\,}{dT}$. Here, $\mH_g(z)=\mH_g(\varnothing,\varnothing;z)$ as in \eqref{Hg2}. Using \eqref{DYZf1}, a more refined structure was proved in \cite{DYZ} that for $g\geq2$, $H_{g}(T)=H_{g}(T(z)):=\mH_{g}(z)$ belong to $\mathbb{Q}[\frac{1}{1-T}]$, with top term 
		\begin{align}
			\frac{c_g}{24^g(5g-3)(5g-5)}\frac{1}{(1-T)^{5g-5}},
		\end{align} 
		where 
		\begin{align}
			c_g=50(g-1)^2 c_{g-1}+\frac{1}{2}\sum_{h=2}^{g-2} c_h c_{g-h},   \qquad\qquad g\geq3,
		\end{align}
		with $c_0=-1,\,c_1=2$, whose generating series
		\begin{align}
			K(Z)=\sum_{g\geq 0} c_g Z^{\frac{1}{2}-\frac{5g}{2}}
		\end{align}
		satisfies the Painlev\'{e} I equation \cite{CGMPS,DYZ,GIKM,IZ}
		\begin{align}
			\frac{d^2 K(Z)}{d Z^2}+\frac{1}{16}K(Z)^2-\frac{1}{16}Z=0.\label{Painleve}
		\end{align}		
		
		As a generalization of \eqref{DYZf1}, based on \eqref{Pand4} and \eqref{D_Pand}, we will obtain recursion formulas of $\mH_{g}(U,V;z)$ (see \eqref{D_H_r2} and \eqref{D_H_r}).
		
		By using a method of Dubrovin-Yang-Zagier \cite{DYZ}, we will prove in the following Theorem~\ref{thm_t_d}, which generalizes part of \cite[Theorem~3]{DYZ} and is a refined version of Theorem~A. 
		
		\begin{theorem}\label{thm_t_d}
			For $U\models t,\,V\models w,\,t,w\in \mathbb{N},\,2g\geq l^*(U)+{l}^*(V)+\delta_{U,\varnothing}+\delta_{V,\varnothing}+1$, the generating series of double Hurwitz numbers with a fixed genus $H_{g}(U,V;T)$ belong to $\mathbb{Q}[\frac{1}{1-T}]$, with top term
			\begin{align}
				\frac{(5g+2l(U)+2l(V)-7)!!\cdot c_g}{24^g(5g-3)!! \#\mathrm{Aut}(U) \#\mathrm{Aut}(V)}\frac{\prod\nolimits_{i}\frac{u_i^{u_i}}{u_i!} \prod\nolimits_{j}\frac{v_j^{v_j}}{v_j!}}{(1-T)^{5g+2l(U)+2l(V)-5}},\label{conj}
			\end{align} 
			where $ \#\mathrm{Aut}(U)=\prod_{i}m_i(U)!$.
		\end{theorem}
		The proof of Theorem~\ref{thm_t_d} is in Section~\ref{Lam}.

		The large degree asymptotics of $H_{g,d}(1^d,1^d,2\,1^{d-2},\allowbreak2\,1^{d-2},\cdots)$ was obtained in \cite{CGMPS,DYZ}. In the following corollary, we generalize these results.		
		\begin{corollary}\label{rmk}
			For any fixed $g\geq 0,\,U\models t,\,V\models w,\,t,w\in \mathbb{N}$, the asymptotics of $H_{g,d}(U\,1^{d-t},\,V\,1^{d-w},2\,1^{d-2},2\,1^{d-2},\cdots)$ is given by
			\begin{align}
				&H_{g,d}(U\,1^{d-t},\,V\,1^{d-w},2\,1^{d-2},\cdots)\sim\frac{\prod\frac{u_i^{u_i}}{u_i!} \prod\frac{v_j^{v_j}}{v_j!}2^{-l^*(U)-l^*(V)} \sqrt{\pi/2}\cdot c_g}{(24\sqrt{2})^g \Gamma(\frac{5g-1}{2}) \#\mathrm{Aut}(U) \#\mathrm{Aut}(V)}  \nonumber  \\
				&  \qquad\qquad\qquad\quad   \times  \left(\frac{4}{e}\right)^{d}d^{2d-5+9g/2+2l(U)+2l(V)-t-w},\quad \,\,\,    \text{as }d\rightarrow \infty.\label{3}
			\end{align}
		\end{corollary}
		We note that a certain universality phenomenon was found by Dubrovin-Yang-Zagier \cite[Theorem~7]{DYZ}. Theorem~\ref{thm_t_d} together with Corollary~\ref{rmk} extends the universality phenomenon of \cite{DYZ}, and is analogous to it. 
		
		Following \cite{DYZ,Hur1}, for $\mu^{(1)},\mu^{(2)}\vdash d$, we define the generating series of double Hurwitz numbers with a fixed degree as
		\begin{align}
			C_d(\mu^{(1)},\mu^{(2)};x):=\sum_{g\geq0} h_{g,d}(\mu^{(1)},\mu^{(2)},2\,1^{d-2},\cdots)x^{2g+2d-l^*(\mu^{(1)})-l^*(\mu^{(2)})-2}.\label{C}
		\end{align}
		The following quadratic recursion for $C_d(x)$ was obtained in~\cite{DYZ} from \eqref{DYZf2}:
		\begin{align}
			(d^3-d^2)C_d(x)=&\sum_{k=1}^{d-1}k(d-k)^2(e^{kx}-2+e^{-kx})C_k(x)C_{d-k}(x),\label{c_C}
		\end{align}		
		where $C_d(x)=C_d(1^d,1^d;x)$. A new proof of the following theorem was given in \cite{DYZ} using the above quadratic recursion \eqref{c_C}.
		
		\vspace{3pt}
		\noindent\bf{Theorem~B}\mdseries{} (Hurwitz~\cite{Hur1}).\textit{ The number $H_{g,d}(1^d,\,1^d)$ for fixed $d$ has the form
			\begin{align}
				H_{g,d}(1^d,\,1^d)=\frac{2}{d!^2}\sum_{1\leq m\leq \binom{d}{2}}b_{d,m}m^{2g+2d-2},
			\end{align}
			where $b_{d,m}$ are integers with $b_{d,\binom{d}{2}}=1$ and $b_{d,m}=0$, for $\binom{d-1}{2}<m<\binom{d}{2}$.}
		\vspace{3pt}
		
		We will use the above-mentioned method from \cite{DYZ} to prove the following Theorem~\ref{tr2}. To state this theorem, it is convenient to introduce the following notations:
		$$\varepsilon_i(z)=\left\{
		\begin{aligned}
			&2\sinh(z)&i\text{ is \text{odd}}\,\,\\
			&2\cosh(z)      &i\text{ is \text{even}}\\
		\end{aligned}
		\right..$$
		\begin{theorem}\label{tr2}
			The generating series of double Hurwitz numbers with a fixed degree $C_d(\mu^{(1)},\mu^{(2)};x)$ with $\mu^{(1)},\,\mu^{(2)}\vdash d$ has an expression of the form
			\begin{align}
				C_d(\mu^{(1)},\mu^{(2)};x)=\left\{
				\begin{aligned}
					&\sum_{k=0}^{\frac{d(d-1)}{2}}\kappa(\mu^{(1)},\mu^{(2)};k)cosh(kx),\,&l(\mu^{(1)})+l(\mu^{(2)})=\text{even}\\
					&\sum_{k=0}^{\frac{d(d-1)}{2}}\kappa(\mu^{(1)},\mu^{(2)};k)sinh(kx),      &l(\mu^{(1)})+l(\mu^{(2)})=\text{odd}\\
				\end{aligned}
				\right.\nonumber
			\end{align}
			and these factors can be determined by the following recursions: for $\mu^{(1)}={(a,\cdots)}\vdash d,\,a\geq2,d\in\mathbb{N}$,
			\begin{align}		 
				C_{d}(&\mu^{(1)},1^d;x)=\sum_{j=1}^{a-1} \sum_{\substack{(P,R)\in\operatorname{Part}^j[\mu^{(1)}]}} \frac{1}{da}\frac{|P^{(0)}|+1}{m_a(\mu^{(1)})}   C_{|P^{(0)}|+1}(1\cup P^{(0)},1^{|P^{(0)}|+1};x)   \nonumber   \\ 
				& \times   (1+m_1(P^{(0)}))\frac{1}{j!}\prod_{i=1}^{j}\mathcal{C}^{P^{(i)}}_{R^{(i)}}\varepsilon_{l(R^{(i)})}(|P^{(i)}| x)C_{|P^{(i)}|}(P^{(i)},1^{|P^{(i)}|};x),     	\label{D_C1}
			\end{align}	
			and \eqref{D_C2} (see Section~\ref{CC}). Moreover, 
			\begin{align}
				\kappa(\mu^{(1)},\mu^{(2)};\tbinom{d}{2})=\frac{2}{z_{\mu^{(1)}} z_{\mu^{(2)}}}.
			\end{align}
		\end{theorem}
		We note that the case of Theorem~\ref{tr2} with $\mu^{(1)}=\mu^{(2)}=\varnothing$ is the situation of Theorem~B.
		
		The following corollary easily follows from Theorem~\ref{tr2}.
		
		\begin{corollary}\label{tr3}
			For all fixed $\mu^{(1)},\mu^{(2)}\vdash d$, the asymptotics of $H_{g,d}(\mu^{(1)},\mu^{(2)},\allowbreak2\,1^{d-2},2\,1^{d-2},\cdots)$
			is given by
			\begin{align}
				H_{g,d}(\mu^{(1)},\mu^{(2)},2\,1^{d-2},2\,1^{d-2},\cdots) \sim&\frac{2}{z_{\mu^{(1)}} z_{\mu^{(2)}}}\binom{d}{2}^{2g+2d-l^*(\mu^{(1)})-l^*(\mu^{(2)})-2},\,\,\text{as }g\rightarrow \infty. 
			\end{align}
		\end{corollary}		
		For the case when $\mu^{(1)}=\mu^{(2)}=1^d$, the asymptotics of $H_{g,d}(1^d,1^d,2\,1^{d-2},\allowbreak2\,1^{d-2},\cdots)$ was given in \cite{DYZ,Hur1,Hur2}. For the case when $\mu^{(2)}=1^d$, the asymptotics was given in \cite{DHR,Y}. See also \cite{L} and references therein for more general case.
		
		This paper is organized as follows: In Section~\ref{recursion}, we derive the equation for $\mH(U,V;x,y)$ and the recursion formulas for $\mH_g(U,V;z)$ and $C_d(\mu^{(1)},\mu^{(2)};x)$. In Section~\ref{Lam}, we give a new proof of Theorem~A and prove Theorem~\ref{thm_t_d}, Corollary~\ref{rmk}. In Section~\ref{CC}, we prove Theorem~\ref{tr2} and Corollary~\ref{tr3}.
		
		\section{Recursion formulas}\label{recursion}		
		In this section, we establish the Pandharipande-type equations of $\mH(U,\varnothing;x,y)$ \eqref{Pand4} and $\mH(U,V;x,y)$ \eqref{D_Pand}, for $U\models t,\,V\models w,\,t,w\in \mathbb{N}$.
		
		Notice that
		\begin{align}
			e^{\sum_{i\geq1}a_i D_{t_i}}\tau(\bt,\bt';x,y)\tau(\bt,\bt';x,y)=&\tau(\bt+\ba,\bt';x,y)\tau(\bt-\ba,\bt';x,y).
		\end{align}
		Here the Hirota derivative \cite{MJD} $D_{t_i}$ is defined for two functions $f(\bt)$ and $g(\bt)$ by
		\begin{align}
			D^n_{t_i}f(\bt)g(\bt):=\left(\frac{\P\,\,}{\P t_i}-\frac{\P\,\,}{\P t'_i}\right)^n f(\bt)g(\bt')\bigg|_{\bt'=\bt}.\label{Hir_de}
		\end{align}
		Equivalently, it can be written as
		\begin{align}
			D^n_{t_i}f(\bt)g(\bt):=\sum_{k=0}^n (-1)^{n-k}\binom{n}{k}\frac{\P^k f}{\P t_i^k}\frac{\P^{n-k} g}{\P t_i^{n-k}}.
		\end{align}
		Let $n=n'=0,\,\bt_+=\bt+\bs,\,\bt'_+=\bt-\bs,\,\bt_-=\bt'+\bs',\,\bt'_-=\bt'-\bs'$ in \eqref{Hir}, we obtain 
		\begin{align}
			&\oint_{C_\infty}\frac{dz}{2\pi i}e^{\xi(2\bs,z)}e^{\sum_{i\geq1}(s_i-\frac{1}{iz_i})D_{t_i}}e^{\sum_{j\geq1}s'_jD_{t'_j}}\tau\cdot\tau       \nonumber     \\
			&\qquad\qquad\qquad=\oint_{C_0}\frac{dz}{2\pi i}e^{\xi(-2\bs',z)}e^{\sum_{i\geq1}s_iD_{t_i}}e^{\sum_{j\geq1}(s'_j+\frac{1}{jz_j})D_{t'_j}}\tau_1\cdot\tau_{-1} , \label{Hir_b}
		\end{align}
		where $\xi(2\bs,z)=\sum_{k\geq1}2s_k z^{k}$ and $\tau=\tau(\bt,\bt';x,y)$. Notice that
		\begin{align}
			e^{\xi(\bt,z)}=&\sum_{k\geq0}h_k(\bt)z^k,\nonumber
		\end{align}
		where
		\begin{align}
			h_k(\bt):=&\sum\limits_{k_1+2k_2+\cdots=k}\frac{t_1^{k_1}}{k_1!}\frac{t_2^{k_2}}{k_2!}\cdots,\qquad k>0,
		\end{align}
		$h_0=1$ are complete symmetry functions.
		For $\bs'=\{s'_1,\,s'_2,\,\cdots\}$, comparing the coefficients of $s'_a$ on both sides (for details, see \cite{Hir2,Miva}), we obtain
		\begin{align}
			& D_{t_1}D_{t'_a}\tau\cdot\tau=2h_{a-1}(\widetilde{D}_{t'})\tau_1\cdot\tau_{-1}  , \label{s'n}
		\end{align}
		where $\widetilde{D}_{t'_j}=\frac{1}{j}D_{t'_j}$,
		By the definition of power sum symmetric functions $p_\lambda(\bt)$ and the relationship between $p_\lambda(\bt)$, $h_{k}(\bt)$ \cite{Mac}, we obtain
		\begin{align}
			h_{a-1}(\widetilde{D}_{t'})= \sum_{\lambda\vdash a-1}\frac{1}{z_\lambda}p_\lambda(\widetilde{D}_{t'})= \sum_{\lambda\vdash a-1}\frac{1}{z_\lambda}\prod_{i}D_{t'_{\lambda_i}}= :\sum_{\lambda\vdash a-1}\frac{1}{z_\lambda}D _{t'_\lambda},\label{eq7}
		\end{align}
		where $z_\lambda=\prod_k m_k(\lambda)! k^{m_k(\lambda)}$ with $m_k(\lambda)$ being the multiplicity of $k$ in $\lambda$. By \eqref{rlt}, \eqref{tau_r}, \eqref{s'n} becomes
		\begin{align}		
			D_{p_1}D_{p'_a}e^{\mH}e^{\mH}=\frac{2y}{a}&\sum_{\lambda\vdash a-1}\frac{1}{\prod_k m_k(\lambda)!}D_{p'_{\lambda}}e^{\mH(\cdots,e^xy)}e^{\mH(\cdots,e^{-x}y)},    \label{eq3}
		\end{align}			
		where $\mH=\mH(\bp,\bp';x,y)$.
		\begin{example}
			When $a=1$, \eqref{eq3} becomes
			\begin{align}
				\frac{\P^2}{\P p_1 \P p'_1} \mH(\bp,\bp';x,y)=ye^{\mH(\bp,\bp';x,e^xy)+\mH(\bp,\bp';x,e^{-x}y)-2\mH(\bp,\bp';x,y)}.\label{eq4}
			\end{align}
			Setting $p_1=p'_1=1,\,p_2=p'_2=p_3=p'_3=\cdots=0$, one gets \eqref{Pand}.
		\end{example}
		Using the similar idea to Okounkov, we will give a equation for $\mH(U,\varnothing;x,y)$ (see \eqref{Pand4}) with partition $U=(u_1\,\cdots)\models t$, for $u_1\geq2,\,t\in\mathbb{N}$.
		\begin{lemma}\label{partial1}
			\begin{align}
				D _{\bp_\lambda}e^{f(\bp)}e^{g(\bp)}= e^{f(\bp)+g(\bp)}&\sum_{j=1}^{l(\lambda)}\sum_{R^{(1)}\cup \cdots\cup R^{(j)}=\lambda}\frac{1}{j!}\bigg(\prod_{k}\frac{m_k(\lambda)!}{\prod_{i=1}^{j} m_k(R^{(i)})!}\bigg)\nonumber\\
				&\qquad\quad\times\prod_{i=1}^{j}\P _{p_{R^{(i)}}}  (f(\bp)+(-1)^{l(R^{(i)})}g(\bp)).        \label{PX}
			\end{align}
			\begin{proof}
				When $l(\lambda)=1$, Lemma~\ref{partial1} obviously holds. By mathematical induction, we assume that for all partitions $\lambda,$ with $l(\lambda)=s-1$, Lemma~\ref{partial1} holds. We next prove Lemma~\ref{partial1} also holds for  $\lambda'=(r)\cup\lambda$. By the definition of the Hirota derivative \eqref{Hir_de}, we have
				\begin{align}
					LHS=  &e^{f(\bp)+g(\bp)}\sum_{j=1}^{l(\lambda)}\sum_{R^{(1)}\cup \cdots\cup R^{(j)}=\lambda}\frac{1}{j!}\Big(\prod_k\tbinom{m_k(\lambda)}{m_k(R^{(1)})\cdots m_k(R^{(j)})} \Big)\Big(\P_{p_r}(f(\bp)  \nonumber  \\
					& -g(\bp))\prod_{i=1}^{j}\P _{p_{R^{(i)}}}  (f(\bp)+(-1)^{l(R^{(i)})}g(\bp))+\sum_{i_0=1}^{j}\P_{p_r}\P _{p_{R^{(i)}}}  (f(\bp)\nonumber\\
					&+(-1)^{l(R^{(i)})+1}g(\bp))\prod_{\substack{i=1,i\neq i_0}}^{j}\P _{p_{R^{(i)}}}  (f(\bp)+(-1)^{l(R^{(i)})}g(\bp))\Big) \nonumber\\
					=&  e^{f(\bp)+g(\bp)}\sum_{\substack{R^{(1)}\cup \cdots\cup R^{(j-a)}\cup\\ \underbrace{(r)\cup\cdots\cup(r)}_{a}=\lambda\\R^{(1)}\cdots R^{(j-a)}\neq (r)}}\frac{1}{j!}\binom{j}{a}\Big(\prod_k\tbinom{m_k(\lambda)}{m_k(R^{(1)})\cdots m_k(R^{(j)})} \Big)\Big(\P_{p_r}(f(\bp)\nonumber  \\
					& -g(\bp))\prod_{i=1}^{j}\P _{p_{R^{(i)}}}  (f(\bp)+(-1)^{l(R^{(i)})}g(\bp))+\sum_{i_0=1}^{j}\P_{p_r}\P _{p_{R^{(i)}}}  (f(\bp)\nonumber\\
					&+(-1)^{l(R^{(i)})+1}g(\bp))\prod_{\substack{i=1,i\neq i_0}}^{j}\P _{p_{R^{(i)}}}  (f(\bp)+(-1)^{l(R^{(i)})}g(\bp))\Big) \nonumber\\
					= & e^{f(\bp)+g(\bp)}\sum_{\substack{R^{(1)'}\cup \cdots\cup R^{(j'-a)'}\cup\\\underbrace{(r)\cup\cdots\cup(r)}_{a}=\lambda'\\R^{(1)'}\cdots R^{(j'-a)'}\neq (r)}}\Big(\prod_{i=1}^{j'}\frac{ m_k(\lambda')!}{\prod_{k,k\neq r}m_k(R^{(i)'})!}\Big)\prod_{i=1}^{j'}\P _{p_{R^{(i)'}}}  (f(\bp)    \nonumber\\
					& +(-1)^{l(R^{(i)'})}g(\bp))\Big(\frac{1}{(j'-1)!}\binom{j'-1}{a-1}\frac{(m_r(\lambda')-1)!}{m_r(R^{(1)'})!\cdots m_r(R^{(j'-a)'})!} \nonumber\\
					& +\frac{1}{j'!}\binom{j'}{a}\sum_{\substack{i=1\\m_r(R'_{i})>0}}^{j'-a}\frac{(m_r(\lambda')-1)!}{m_r(R^{(1)'})!\cdots (m_r(R^{(i)'})-1)!\cdots m_r(R^{(j'-a)'})!}\Big)	    \nonumber \\
					= & RHS. 			
				\end{align}
			\end{proof}
		\end{lemma}		
		
		In \eqref{eq3}, using Lemma~\ref{partial1} and noticing
		\begin{align}
			H_{g,d}(\mu^{(1)},\mu^{(2)},2\,1^{d-2},2\,1^{d-2},\cdots)=H_{g,d}(\mu^{(2)},\mu^{(1)},2\,1^{d-2},2\,1^{d-2},\cdots),\label{sym}
		\end{align}
		we have 
		\begin{align}\label{n_p}		
			\P_{p_{u_1}}&\P_{p_{1'}}\mH=ye^{\mH(\cdots,e^x y)+\mH(\cdots,e^{-x}y)-2\mH}\sum_{j=1}^{u_1-1}\sum_{|R^{(1)}|+ \cdots  |R^{(j)}|=u_1-1}\frac{1}{ u_1j!}\prod_{i=1}^{j}        \nonumber\\
			&\times\big(\prod_{k}\frac{1}{m_k(R^{(i)})!}\big)\P{p_{R^{(i)}}}\big(\mH(\cdots,e^x y)+(-1)^{l(R^{(i)})}\mH(\cdots,e^{-x}y)\big),
		\end{align}	
		where by an abuse of notion
		\begin{align}
			\mH=\mH(\bp,\,\bp',x,y),\quad \mH(\cdots,e^x y)=\mH(\bp,\,\bp',x,e^x y).
		\end{align}		
		Combining with \eqref{eq4}, we obtain
		\begin{align}\label{n_p2}		
			\P_{p_{u_1}}\P_{p_{1'}}\mH=\P_{p_{1}}\P_{p_{1'}}\mH&\sum_{j=1}^{u_1-1}\sum_{|R^{(1)}|+ \cdots  |R^{(j)}|=u_1-1}\frac{1}{ u_1j!}\prod_{i=1}^{j}\big(\prod_{k}\frac{1}{m_k(R^{(i)})!}\big)        \nonumber\\
			&\times\P{p_{R^{(i)}}}\big(\mH(\cdots,e^x y)+(-1)^{l(R^{(i)})}\mH(\cdots,e^{-x}y)\big).
		\end{align}			
		Applying $\P {p_{U\setminus u_1}}$ with $U\setminus u_1=(u_2,\cdots,u_{l(U)})$ to both sides of \eqref{n_p}, we have	
		\begin{align}
			\P&_{p_{U}}\P_{p_{1'}}\mH=\sum_{j=1}^{u_1-1} \sum_{(P,R)\in\operatorname{Part}^j[U]}\frac{1}{ u_1j!}\P_{p_{1}}\P_{p_{P^{(0)}}}\P_{p_{1'}}\mH\prod_k\tbinom{m_k(U\setminus u_1)}{m_k(P^{(0)})\cdots m_k(P^{(j)})} \prod_{i=1}^{j}\nonumber\\
			& \times\frac{1}{\prod_k m_k(R^{(i)})!}\P_{p_{R^{(i)}}}\P_{p_{P^{(i)}}}\left(\mH(\cdots,e^x y)+(-1)^{l(R^{(i)})}\mH(\cdots,e^{-x}y)\right),   \label{eq2}
		\end{align}		
		where $P^{(i)}$ allow to be empty set. Setting $p_1=p'_1=1,\,p_2=p'_2=p_3=p'_3=\cdots=0$, we obtain the equation of $\mH(U,\varnothing;x,y)$ (as in \eqref{Huvxy}) \eqref{Pand4}.
		
		\begin{example}
			When $u_1=2$, the partition $U=2^a,\,a\geq1$. Let $\nu=\varnothing$, then \eqref{Pand4} becomes
			\begin{align}		 
				\mH^{[0,1]}(2^{a},\varnothing;x,y)= \frac{1}{2a}&\sum_{a_0+a_1=a-1}\mH^{[1,1]}(2^{a_0},\varnothing;x,y)\nonumber\\
				&\times  \left(\mH^{[1,0]}(2^{a_1},\varnothing;x,e^x y)-\mH^{[1,0]}(2^{a_1},\varnothing;x,e^{-x}y)\right)  .\nonumber
			\end{align}	
		\end{example}
		
		We next will give another equation of the generating series of $\mH(U,V;x,y)$ with partitions $U=(u_1,\cdots)\models t,\,V=(v_1,\cdots)\models w,\,t,w\in\mathbb{N}$. We compare the coefficients of $s_{u_1-1}s'_{v_1}$ on both sides of \eqref{Hir_b}. By \eqref{rlt}, \eqref{tau_r} and \eqref{eq7}, we obtain
		\begin{align}
			D_{p_{u_1}}D_{p'_{v_1}}e^{\mH}e^{\mH}=\sum_{\sigma\vdash v_1-1}&\frac{(u_1-1)y}{v_1\prod_l m_l(\sigma)!}D_{p'_\sigma}D_{p_{u_1-1}}e^{\mH(\cdots,e^xy)}e^{\mH(\cdots,e^{-x}y)}       \nonumber\\
			&  +\sum_{\substack{\lambda\vdash u_1,\,\lambda\neq u_1}}\frac{(-1)^{l(\lambda)}}{\prod_k m_k(\lambda)!}D_{p_\lambda}D_{p'_{v_1}}e^{\mH}e^{\mH}.        \label{eq8} 
		\end{align}
		Similarly to derivation of \eqref{Pand4}, we obtain \eqref{D_Pand}.
		\begin{example}
			When $u_1=v_1=2$, the partition $\mu=2^a,\,\nu=2^b,\,a,\,b\geq1$, then \eqref{D_Pand} becomes
			\begin{align}		 
				\mH(2^{a},2^{b};x,y)&= \frac{1}{4ab}\sum_{\substack{a_0+a_1=a-1\\b_0+b_1=b-1}}\mH^{[1,1]}(2^{a_0},2^{b_0};x,y)\big(\mH^{[1,1]}(2^{a_1},2^{b_1};x,e^x y)\nonumber\\
				&+\mH^{[1,1]}(2^{a_1},2^{b_1};x,e^{-x}y)\big)+\frac{1}{4ab}\sum_{\substack{a_0+a_1+a_2=a-1\\b_0+b_1+b_2=b-1}}\mH^{[1,1]}(2^{a_0},2^{b_0};x,y)  \nonumber\\
				&\times \left(\mH^{[0,1]}(2^{a_1},2^{b_1};x,e^x y)-\mH^{[0,1]}(2^{a_1},2^{b_1};x,e^{-x}y)\right)\nonumber\\
				& \times\left(\mH^{[1,1]}(2^{a_2},2^{2b_2};x,e^x y)-\mH^{[1,0]}(2^{a_2},2^{b_2};x,e^{-x}y)\right).\nonumber
			\end{align}	
		\end{example}

		\section{Structure with a fixed genus}\label{Lam}		
		In this section, we establish the recurrence formula of $\mH_{g}(U,\varnothing;z)$ \eqref{D_H_r2} and $\mH_{g}(U,V;z)$ \eqref{D_H_r}. Then, we give the proof of Theorem~\ref{thm_t_d}, Theorem~\ref{h}, Corollary~\ref{rmk}.
		
		Let $z=x^2y$. Comparing coefficients of $x^{2g-l^*(U)-2}$ on both sides of \eqref{Pand4}, we give the following recursion formula of $\mH_{g}(U,\varnothing;z)$, with $U=(u_1,\dots)\models t$
		\begin{align}  
			\mH^{(1)}_{g}(U,\varnothing ;z)=&\sum_{j=1}^{u_1-1}\sum_{\substack{(P,R)\in\operatorname{Part}^j[U]}}\sum_{\substack{g_0,\cdots,\,g_{j},\,k_1,\cdots,k_{j}\geq 0\\2g_0+\sum_{i=1}^{j}2g_{i}+2k_{i}+\\l(R^{(i)})+\delta_{\text{odd}}^{l(R^{(i)})}=2g+2j}} \frac{1}{u_1j!m_{u_1}(U)}     \nonumber \\
			&\Big(\sum_{l_0=0}^{1}[d^{l_0}]\big(\tbinom{d-|P^{(0)}|}{1}\big)\mH_{g_0}^{(1+l_0)}(P^{(0)},\varnothing;z)\Big)\prod_{i=1}^{j}\Big(\frac{2\mathcal{C}^{P^{(i)}}_{R^{(i)}}}{(2k_i+\delta^{l(R^{(i)})}_{\text{odd}})!}      \nonumber \\
			&\sum_{l_i=0}^{m_1(R^{(i)})}[d^{l_i}]\big(\tbinom{d-|P^{(i)*}|}{m_1(P^{(i)})}\big)\mH_{g_i}^{(l_i+2k_i+\delta_{\text{odd}}^{l(R^{(i)})})}(P^{(i)*},\varnothing;z)\Big),  \label{D_H_r2}
		\end{align}
		where 
		$$
		\mH^{(\ell)}_{g}(U,V;z):=D^\ell\mH_{g}(U,V;z)=\sum\limits_{d} d^{\ell} h_{g,d}(U\,1^{d-t},\,V\,1^{d-w},\,2\,1^{d-2},\cdots)z^d,
		$$
		with $D$ as in \eqref{DH}, $[d^l]\big(f(d)\big)$ represent the coefficient of $d^l$ in polynomial $f(d)$ and $\delta^{l(R)}_{\text{odd}}=\left\{
		\begin{aligned}
			&1\,\,&l(R^{(i)})\text{ is \text{odd}}\\
			&0      &l(R^{(i)})\text{ is \text{even}}\\
		\end{aligned}
		\right.$.
		
		Let's discuss several specific example in the following Corollary.
		\begin{corollary}
			When $u_1=2$, the partition $U=2^a,\,a\geq1$, then \eqref{D_H_r2} becomes:
			\begin{align}
				\mH_{g}^{(1)}(2^{a},\varnothing;z)&=\sum_{\substack{a_0,a_1\geq0,\\a_0+a_1=a-1}}\sum_{\substack{g_0,g_1,k_1\geq0\\g_0+g_1+k_1=g}}\frac{1}{a(2k_1+1)!} \big(\mH_{g_0}^{(2)}(2^{a_0},\varnothing;z)-2a_0\label{S_H_r_2}\\
				&\,\,\,\mH_{g_0}^{(1)}(2^{a_0},\varnothing;z) \big) \big(\mH_{g_1}^{(2k_1+2)}(2^{a_1},\varnothing;z)-2a_1\mH_{g_1}^{(2k_1+1)}(2^{a_1},\varnothing;z) \big)   ,	\nonumber
			\end{align}
			and for $a=0$ or $1$,
			\begin{align}
				\mH_{g}^{(1)}(2,\varnothing;z)=&\sum_{\substack{g_0,g_1,k_1\geq0\\g_0+g_1+k_1=g}}\frac{1}{(2k_1+1)!} \mH_{g_0}^{(2)}(\varnothing,\varnothing;z) \mH_{g_1}^{(2k_1+2)}(\varnothing,\varnothing;z),	\nonumber\\
				\mH_{g}^{(1)}(2^{2},\varnothing;z)=&\sum_{\substack{g_0,g_1,k_1\geq0\\g_0+g_1+k_1=g}}\frac{1}{2(2k_1+1)!} \Big(\big(\mH_{g_0}^{(2)}(2,\varnothing;z)-2\mH_{g_0}^{(1)}(2,\varnothing;z) \big) \nonumber\\
				\times\mH_{g_1}^{(2k_1+2)}(&\varnothing,\varnothing;z)+\mH_{g_0}^{(2)}(\varnothing,\varnothing;z)\big(\mH_{g_1}^{(2k_1+2)}(2,\varnothing;z)-2\mH_{g_1}^{(2k_1+1)}(2,\varnothing;z) \big)\Big).	\nonumber
			\end{align}
		\end{corollary}
		Comparing coefficients of $x^{2g-l^*(U)-l^*(V)-2}$ on both sides of \eqref{D_Pand}, we obtain the following recursion formula for $\mH_{g}(U,V;z)$, with $U=(u_1,\dots)\models t,\,V=(v_1,\dots)\models w,\,t,w\in\mathbb{N}$.
		\begin{align}      
			&\mH_{g}(U,V;z) =\sum_{r=1}^{v_1-1}\sum_{i_0=1}^{r+1}\sum_{\substack{P\in\operatorname{Split}_{i_0}^{r+\delta^{i_0}_{r+1}}[U]\\(\overline{P},S)\in\operatorname{Part}^{r+\delta^{i_0}_{r+1}}[V]   \\ }} \sum_{\substack{g_0,\cdots, g_{r+\delta^{i_0}_{r+1}}, p_1,\cdots,p_{r+\delta^{i_0}_{r+1}}\geq 0\\\sum_i 2g_{i}+\sum_i 2p_{i}+\sum_i l(S^{(i)})\\+\delta^i_{i_0}+\delta_{\text{odd}}^{l(S^{(i)})+\delta^i_{i_0}}=2g+
					2r}}  \nonumber  \\
			&  \times \frac{(u_1-1)m_{(u_1-1)}(P^{(i_0)})}{2v_1  r! m_{u_1}(U)m_{v_1}(V)}\bigg(\sum_{l_0=0}^{2}[d^{l_0}]\Big(\tbinom{d-|P^{(0)}|}{1}\tbinom{d-|\overline{P}^{(0)}|}{1}\Big)\mH_{g_0}^{(l_0)}(P^{(0)} ,\overline{P}^{(0)} ;z)\bigg)\nonumber   \\
			& \times\bigg(\prod_{i=1}^{r+\delta^{i_0}_{r+1}}\frac{2\mathcal{C}^{\overline{P}^{(i)} }_{S^{(i)} }}{(2p_{i }+\delta^{l(S^{i})+\delta^i_{i_0}}_{\text{odd}})!}\sum_{q_i=0}^{m_1(P^{(i)}\cup \overline{P}^{(i)})}[d^{q_i}]\Big(\tbinom{d-|P^{(i)*}|}{m_1(P^{(i)})}\tbinom{d-|\overline{P}^{(i)*}|}{m_1(\overline{P}^{(i)})}\Big) \nonumber\\
			&\times\mH_{g_{i }}^{(q_{i }+2p_{i }+\delta_{\text{odd}}^{l(S^{(i)})+\delta^i_{i_0}})}(P^{(i)*} ,\overline{P}^{(i)*} ;z)\bigg)-\sum_{w=1}^{[\frac{u_1}{2}]+1} \sum_{j_0=1}^{w}\sum_{\substack{(P,R)\in\overline{\operatorname{Part}}^{w}[U]\\\overline{P}\in\overline{\operatorname{Split}}_{j_0}^{w}[V]}}\sum_{\substack{g_1, g_2,\cdots, g_{w}\geq 0\\\sum_{j=1}^{w}2g_{j}+l(R^{(j)})\\=2g+2w-1}}\nonumber\\
			& \frac{2^{w}  m_{(v_1)}(\overline{P}^{(j_0)})}{2w! m_{u_1}(U)m_{v_1}(V)}\prod_{\substack{j=1}}^{w} \sum_{q_j =0}^{m_1(P^{(j)})}[d^{q_j}]\Big(\tbinom{d-|P^{(j)*}|}{m_1(P^{(j)})}\Big)\mH_{g_j }^{(q_j )}(P^{(j)*}, \overline{P}^{(j)} ;z).    \label{D_H_r}
		\end{align}
		\begin{corollary}
			When $u_1=v_1=2$, the partitions $U=2^a,\,V=2^b$, $a,b\geq1$, then \eqref{D_H_r} becomes:
			\begin{align}
				\mH_{g}&(2^{a},2^{b};z)=\sum_{\substack{a_0,a_1,a_2,b_0,b_1,b_2\geq0,\\a_0+a_1+a_2=a-1\\b_0+b_1+b_2=b-1}}\sum_{\substack{g_0,g_1,g_2,p_1,p_2\geq0\\g_0+g_1+g_2\\+p_1+p_2=g}}  \frac{1}{ab(2p_1+1)!(2p_2+1)!}  \nonumber    \\
				&\times\left(\mH_{g_0}^{(2)}(2^{a_0},2^{b_0};z)-2(a_0+b_0)\mH_{g_0}^{(1)}(2^{a_0},2^{b_0};z)+4a_0b_0\mH_{g_0}(2^{a_0},2^{b_0};z) \right)\nonumber\\
				&\times\left(\mH_{g_1}^{(2p_1+2)}(2^{a_1},2^{b_1};z)-2b_1\mH_{g_1}^{(2p_1+1)}(2^{a_1},2^{b_1}:z)\right)\Big(\mH_{g_2}^{(2p_2+2)}(2^{a_2},2^{b_2};z)\nonumber \\
				&-2a_2\mH_{g_2}^{(2p_2+1)}(2^{a_2},2^{b_2};z)\Big)+\sum_{\substack{a_0,a_1,b_0,b_1\geq0\\a_0+a_1=a-1\\b_0+b_1=b-1}}\sum_{\substack{g_0,g_1,p_1\geq0\\g_0+g_1+p_1=g}}\Big(\mH_{g_1}^{(2p_1+2)}(2^{a_1},2^{b_1};z)	\nonumber \\	
				&-2(a_1+b_1)\mH_{g_1}^{(2p_1+1)}(2^{a_1},2^{b_1};z)+4a_1b_1\mH_{g_1}^{(2p_1)}(2^{a_1},2^{b_1};z)\Big)\frac{1}{2ab(2p_1)!} 			\nonumber\\
				&\times\Big(\mH_{g_0}^{(2)}(2^{a_0},2^{b_0};z)-2 (a_0 +b_0)\mH_{g_0}^{(1)}(2^{a_0},2^{b_0};z)+4a_0b_0 \mH_{g_0}(2^{a_0},2^{b_0};z) \Big),\nonumber
			\end{align}
			and for $a=2,\,b=1$,
			\begin{align}
				\mH_{g}(2^{2},&2;z)=\sum_{\substack{g_0,g_1,p_1\geq0\\g_0+g_1+p_1=g}}\frac{1}{4(2p_1)!} \big(\big(\mH_{g_0}^{(2)}(2,\varnothing;z)-2\mH_{g_0}^{(1)}(2,\varnothing;z) \big)\mH_{g_1}^{(2p_1+2)}(\varnothing,\varnothing;z) \nonumber\\
				&+\mH_{g_0}^{(2)}(\varnothing,\varnothing;z)\big(\mH_{g_1}^{(2p_1+2)}(2,\varnothing;z)-2\mH_{g_1}^{(2p_1+1)}(2,\varnothing;z) \big)\big)+\sum_{\substack{g_0,g_1,g_2,p_1,p_2\geq0\\g_0+g_1+g_2\\+p_1+p_2=g}}  \nonumber\\
				&\frac{1}{2(2p_1+1)!(2p_2+1)!}\Big(\big(\mH_{g_0}^{(2)}(2,\varnothing;z)-2\mH_{g_0}^{(1)}(2,\varnothing;z) \big)\mH_{g_1}^{(2p_1+2)}(\varnothing,\varnothing;z)\nonumber\\
				&\times\mH_{g_2}^{(2p_2+2)}(\varnothing,\varnothing;z)+\mH_{g_0}^{(2)}(\varnothing,\varnothing;z)\mH_{g_1}^{(2p_1+2)}(2,\varnothing;z)\mH_{g_2}^{(2p_2+2)}(\varnothing,\varnothing;z)\nonumber\\
				&+\mH_{g_0}^{(2)}(\varnothing,\varnothing;z)\mH_{g_1}^{(2p_1+2)}(\varnothing,\varnothing;z)\big(\mH_{g_2}^{(2p_2+2)}(2,\varnothing;z)-2\mH_{g_2}^{(2p_2+1)}(2,\varnothing;z) \big)\Big)
				.	\nonumber
			\end{align}
		\end{corollary}
		\begin{remark}
			Throughout the recursion, the multiplicity $m_{u_1}(U)$ is reduced by one at every step. Consequently, after finitely many steps, we obtain an expression involving only $\mH(\varnothing,\varnothing;z)$, which has already been calculated in \cite{DYZ}.
		\end{remark}
		\begin{remark}
			The recursion formula of $\mH_g(U,\varnothing;z)$ \eqref{D_H_r2} requires one integration. Although the recursion formula of $\mH_g(U,V;z)$ \eqref{D_H_r} has a complicated form, \eqref{D_H_r} does not require integration.
		\end{remark}
		
		We can obtain the value of all the generating series of double Hurwitz numbers from the value of $H_g(\varnothing,\varnothing;T)=H_g(\varnothing,\varnothing;T(z)):=\mH_g(\varnothing,\varnothing;z)$ \cite{DYZ}, by recursion formulas \eqref{D_H_r2} and \eqref{D_H_r}. But we still need following method to ensure that $H_g(U,\varnothing;T)$ do not contain term $\log(1-T)$, when we integrate $H^{(1)}_g(U,\varnothing;T)$ obtained by \eqref{D_H_r2}.
		
		For simplicity, we will use a method from \cite{GJV} (see also \cite{GJ4,GJ3,V}). Setting the $p_1=1,\,p_2=p_3=\cdots=0$ in $\mH^{GJ}_g(\bp;z)$ reduce is to $\mH_g(\varnothing,\varnothing;z)$. Goulden, Jackson and Vakil \cite{GJV} calculated the value of $\mH_g(\varnothing,\varnothing;z)$ for $g\leq3$, in this method. Similar to what we did in previous Section, we apply a series of derivatives with respect to $p_2,\,p_3,\,\cdots$, and then set $p_1 = 1$, $p_2 = p_3 = \cdots = 0$. We obtain $H_{g}(U,\varnothing;T)$ belong to $\mathbb{Q}[T, \frac{1}{1-T}]$. Combining this result with the recursion formula \eqref{D_H_r}, we show that $H_{g}(U,V;T)$ also lie in $\mathbb{Q}[T, \frac{1}{1-T}]$.
		
		The generating series of single Hurwitz numbers $\mH^{GJ}_g(\bp;z)$ are
		\begin{align}\label{GJV2}
			\mH^{GJ}_g(\bp;z)&:= \sum\limits_{d}\sum\limits_{\lambda\vdash d} h_{g,d}(\lambda,1^d,\,2\,1^{d-2},2\,1^{d-2},\cdots)z^d p_\lambda     \nonumber\\
			&=:\sum\limits_{t}\sum_{U \models t}\mH_g(p_1,U,\varnothing;z)p_\mu,
		\end{align}
		whose relationship with $\mH_g(\mu,\varnothing;z)$ is
		\begin{align}
			\mH_g(U,\varnothing;z)&=\mH_g(p_1,U,\varnothing;z)|_{\substack{p_1=1}}=\sum_{d}   h_{g,d}(\mu\,1^{d-|\mu|},\,1^d,2\,1^{d-2},2\,1^{d-2},\cdots) z^d.     \nonumber
		\end{align}
		The explicit forms of $\mH^{GJ}_g(\bp;z)$ are known:
		
		For $g=0,\,1$, \cite{GJ4,GJ3,V}
		\begin{align}
			&\mH^{GJ}_1(\bp;z)=\frac{1}{24}\big(\log(1-\phi_1(\bp;s))^{-1}-\phi_0(\bp;s)\big),\label{GJV4}\\
			&\frac{\P }{\P p_k}\mH^{GJ}_0(\bp;z)=\frac{k^{k-2}}{k!}s^k-\frac{k^{k-1}}{k!}\sum_{i\geq1}\frac{i^{i+1}}{i!}p_i\frac{s^{k+i}}{k+i},\label{GJV5}
		\end{align}		
		where
		\begin{align}\label{psi}
			\phi_j(\bp;s):=\sum_{n\geq1}\frac{n^{n+j}}{n!}p_n s^n,\quad \qquad s(z)=z e^{\phi_0(\bp;s(z))}.
		\end{align}		
		For $g\geq2$, \cite{GJV}
		\begin{align}
			\mH^{GJ}_g(\bp;z)=&\sum_{e=2g-1}^{5g-5}\frac{1}{(1-\phi_1(\bp;s))^e}\sum_{n=e-1}^{e+g-1}\sum_{\substack{\theta\models n,\,l(\theta)\\=e-2g+2}}             \nonumber          \\
			&\quad      \frac{\phi_{\theta_1}(\bp;s)\phi_{\theta_2}(\bp;s)\cdots}{\prod_{i}m_i(\theta)!} (-1)^{k}<\tau_{\theta_1}\tau_{\theta_2}\cdots\lambda_{k}>_g.        \label{GJV1}
		\end{align}
		The intersection number $<\tau_{\theta_1}\tau_{\theta_2}\cdots\lambda_k>_g$, for non-negative integers $\theta_1,\,\theta_2,\,\cdots$, is
		\begin{align}\label{Int}
			<\tau_{\theta_1}\tau_{\theta_2}\cdots\lambda_k>_g=\int_{\overline{\mathcal{M}}_{g,l(\theta)}}\psi_1^{\theta_1}\psi_2^{\theta_2}\cdots\lambda_{k},
		\end{align}
		where $\overline{\mathcal{M}}_{g,n}$ denotes the Deligne-Mumford moduli space of stable algebraic curves of genus $g$ with $n$ distinct marked points. And $\psi_j$, $\lambda_k$ are Chern classes on $\overline{\mathcal{M}}_{g,n}$ of codimension $1$ and $k$, respectively, with $1 \leq j \leq n$, $0 \leq k \leq g$, and $\lambda_0 = 1$.
		
		\begin{proof}[\bf{Proof of Theorem~A}]
			We apply $\P_{p_U}$ with partition $U=(u_1,\cdots)$ to both sides of equations \eqref{GJV4}, \eqref{GJV5} and \eqref{GJV1}, then set $p_1 = 1$, $p_2 = p_3 = \cdots = 0$. Notice that, for all $g$,
			\begin{align}
				&\P_{p_U}\mH^{GJ}_g(\bp;z)\Big|_{\substack{p_1=1,\\p_2= p_3=\cdots=0}}= \#\mathrm{Aut}(U)\mH_{g}(U,\varnothing;z),
			\end{align}
			where $\#\mathrm{Aut}(U)=m_2(U)!m_3(U)!\cdots m_n(U)!$. We get
			\begin{align}
				&\mH_{g}(U,\varnothing;z)=\sum_{e=2g-1}^{5g-5}\sum_{n=e-1}^{e+g-1}\sum_{\substack{\theta \models n,\,l(\theta)\\=e-2g+2}}(-1)^{k}\frac{<\tau_{\theta_1}\tau_{\theta_2}\cdots\lambda_{k}>_g}{\prod_{i}m_i(U)!m_i(\theta)!}\nonumber\\
				&\qquad\qquad\qquad\qquad\P_{p_U}\frac{\phi_{\theta_1}(\bp;s)\phi_{\theta_2}(\bp;s)\cdots}{(1-\phi_1(\bp;s))^e}\Big|_{\substack{p_1=1,p_2= p_3=\cdots=0}}, \qquad\quad g\geq2   \label{psiH2}\\
				&\mH_{1}(U,\varnothing;z)=\frac{1}{24}\frac{1}{\prod_{i}m_i(U)!}\P_{p_{U\setminus (u_1)}}\big(\frac{1}{1-\phi_1(\bp;s)}\frac{\P}{\P p_{u_1}}\phi_1(\bp;s)         \nonumber        \\
				&\qquad\qquad\qquad\qquad-\frac{\P}{\P p_{u_1}}\phi_0(\bp;s)\big)|_{\substack{p_1=1,p_2= p_3=\cdots=0}},    \label{psiH1}\\
				&\mH_{0}(U,\varnothing;z)=\frac{1}{\prod_{i}m_i(U)!}\P_{p_{U\setminus u_1}}\Big(\frac{u_1^{u_1-2}}{u_1!}s^{u_1} -\frac{u_1^{u_1-1}s^{u_1+1}}{(u_1+1)!}\Big)\Big|_{\substack{p_1=1,p_2= p_3=\cdots=0.}}\label{psiH0}
			\end{align}
			
			Subsequently, we discuss the right hand sides of \eqref{psiH2}, \eqref{psiH1} and \eqref{psiH0}. Based on equations \eqref{psi}, we have
			\begin{align}
				\frac{\P s}{\P p_{i}}&=\frac{i^{i}}{i!}\frac{s^{i+1}}{(1-\phi_1(\bp;s))},\label{psi4}\\
				\frac{\P}{\P p_{i}}\phi_j(\bp;s)&=\frac{i^{i+j}}{i!}s^{i}+\frac{\phi_{j+1}(\bp;s)}{1-\phi_1(\bp;s)}\frac{i^{i}}{i!}s^{i},\label{psi2}\\
				\frac{\P}{\P p_{i}}\frac{1}{(1-\phi_1(\bp;s))^a}=&(\frac{i}{(1-\phi_1(\bp;s))^{a+1}}+\frac{\phi_2(\bp;s)}{(1-\phi_1(\bp;s))^{a+2}})\frac{a\cdot i^{i}}{i!}s^{i}.\label{psi3}
			\end{align}		
			Here, $s=s(z)$. Setting $T:=s\big|_{\substack{p_1=1,p_2= p_3=\cdots=0}}$ and taking $p_1=1,\,p_2=p_3=\cdots=0$ in \eqref{psi}, we get 
			\begin{align}\label{psiL}
				\phi_j(\bp;s)\big|_{\substack{p_1=1,p_2= p_3=\cdots=0}}=T,\qquad\qquad T=z e^{T}.
			\end{align}		
			Similarly, taking $p_1=1,\,p_2=p_3=\cdots=0$ in \eqref{psi4}, \eqref{psi2}, we obtain
			\begin{align}
				\frac{\P s}{\P p_{i}}\Big|_{\substack{p_1=1,\\p_2= p_3=\cdots=0}}=&\frac{i^{i}}{i!}\frac{T^{i+1}}{1-T},\label{psi7}\\
				\frac{\P}{\P p_{i}}\phi_j(\bp;s)\Big|_{\substack{p_1=1,\\p_2= p_3=\cdots=0}}=&\frac{i^{i+j}}{i!}T^{i}+\frac{T^{i+1}}{1-T}\frac{i^{i}}{i!},\label{psi5}\\
				\frac{\P}{\P p_{i}}\frac{1}{(1-\phi_1(\bp;s))^a}\Big|_{\substack{p_1=1,\\p_2= p_3=\cdots=0}}=&\frac{aT^{i}}{(1-T)^{a+1}}\frac{i^{i+1}}{i!}+\frac{aT^{i+1}}{(1-T)^{a+2}}\frac{i^{i}}{i!}.\label{psi6}
			\end{align}
			Substituting \eqref{psi7}, \eqref{psi5}, \eqref{psi6} into \eqref{GJV4}, \eqref{GJV5}, \eqref{GJV1}, \eqref{psiH2}, \eqref{psiH1} and \eqref{psiH0}, we conclude that $H_{g}(U,\varnothing;T)$ belong to $\mathbb{Q}[T, \frac{1}{1-T}]$, unless $H_{1}(\varnothing,\varnothing;T)=-\frac{T+\log(1-T)}{24}$. Combined with recursion formula for $H_{g}(U,V;T)$ \eqref{D_H_r}, which involves no integration, we further show that $H_{g}(U,V;T)$ also lie in $\mathbb{Q}[T, \frac{1}{1-T}]$, unless $H_{1}(\varnothing,\varnothing;T)=-\frac{T+\log(1-T)}{24}$.
		\end{proof}	
		
		We now give some formulas which will be used in the proofs that follow. When $g=0$, using \eqref{psiH0} and \eqref{psi7}, we obtain for $n\neq m\geq1$,
		\begin{align}
			H_{0}((n),\varnothing;T)=&\frac{\P }{\P p_{n}}\mH^{GJ}_0(\bp;z)\big|_{\substack{p_1=1,\\p_2= p_3=\cdots=0}}=\frac{n^{n-2}}{n!}T^{n}-\frac{n^{n-1}}{(n+1)!}T^{n+1},\label{psiH01Top}\\
			H_{0}((n,m),\varnothing;T)=&\frac{\P^2 }{\P p_{n}p_{m}}\mH^{GJ}_0(\bp;z)\big|_{\substack{p_1=1,p_2= p_3=\cdots=0}}=\frac{n^{n}m^{m}}{n!m!(n+m)}T^{n+m},\label{psiH02Top}\\
			H_{0}((n,n),\varnothing;T)=&\frac{1}{2}\frac{\P^2 }{\P p_{n}p_{n}}\mH^{GJ}_0(\bp;z)\big|_{\substack{p_1=1,p_2= p_3=\cdots=0}}=\frac{n^{2n}}{4n\cdot n!^2}T^{2n}.\label{psiH02Top2}
		\end{align}
		
		Using the result of $\mH^{GJ}_g(\bp;z)$, we obtain the structure of $H_{g}(U,\varnothing;T)$. Then, we study a more refined structure of $H_{g}(U,V;T)$ by the recursion formulas of $H_{g}(U,V;T)$. 
		
		Before discussing further properties of $H_{g}(U,V ;T)$, we introduce a new recurrence formula, which will play a key role in the subsequent analysis. For any partition $V=(v_1,\cdots,v_{l(V)})$, applying $\P_{p_A}\P_{p'_V},\,A=U\setminus u_1$ rather than $\P_{p_A}$ to both sides of \eqref{n_p} and the setting $p_1=p'_1=1,\,p_2=p'_2=p_3=p'_3=\cdots=0$, we obtain
		\begin{align}		 
			\mH(U,V;x,y)= & \sum_{j=1}^{u_1-1} \sum_{\substack{(P,R)\in\operatorname{Part}^j[U]\\\cup_{i=0}^j \overline{P}^{(i)}=V}}\frac{1}{u_1m_{u_1}(U)}\mH^{[1+m_1(P^{(0)}),1]}(P^{(0)*},\overline{P}^{(0)};x,y)    \nonumber   \\ 
			&\frac{1}{j!}\prod_{i=1}^{j}\mathcal{C}^{P^{(i)}}_{R^{(i)}}\cdot \big(\mH^{[m_1(P^{(i)}),0]}(P^{(i)*},\overline{P}^{(i)};x,e^x y)\nonumber\\
			&+(-1)^{l(R^{(i)})}\mH^{[m_1(P^{(i)}),0]}(P^{(i)*},\overline{P}^{(i)};x,e^{-x}y)\big), \label{Pand5}
		\end{align}	
		Comparing coefficients of $x^{2g-l^*(U)-l^*(V)-2}$ on both sides of \eqref{Pand5}, we obtain the following equation for $\mH_{g}(U,V;z)$
		\begin{align}  
			&\mH^{(1)}_{g}(U,V ;z)-|V|\mH_{g}(U,V ;z)=\sum_{j=1}^{u_1-1}\sum_{\substack{(P,R)\in\operatorname{Part}^j[U]\\\cup_{i=0}^j \overline{P}^{(i)}=V}}\sum_{\substack{g_0,\cdots,\,g_{j},\,k_1,\cdots,k_{j}\geq 0\\2g_0+\sum_{i=1}^{j}2g_{i}+2k_{i}+\\l(R^{(i)})+\delta_{\text{odd}}^{l(R^{(i)})}=2g+2j}}    \nonumber  \\
			&\frac{1}{u_1j!m_{u_1}(U)}\bigg(\sum_{l_0=0}^{2}[d^{l_0}]\Big(\tbinom{d-|P^{(0)}|}{1}\tbinom{d-|\overline{P}^{(0)}|}{1}\Big)\mH_{g_0}^{(l_0)}(P^{(0)} ,\overline{P}^{(0)} ;z)\bigg)\prod_{i=1}^{j}       \nonumber \\
			&\bigg(\frac{2\mathcal{C}^{P^{(i)}}_{R^{(i)}}}{(2k_i+\delta^{l(R^{(i)})}_{\text{odd}})!}\sum_{l_i=0}^{m_1(R^{(i)})}[d^{l_i}]\Big(\tbinom{d-|P^{(i)*}|}{m_1(P^{(i)})}\Big)\mH_{g_i}^{(l_i+2k_i+\delta_{\text{odd}}^{l(R^{(i)})})}(P^{(i)*},\overline{P}^{(i)};z)\bigg).  \label{D_H_r3}
		\end{align}
		\begin{corollary}
			For example, let $U=2^{a+1},\,V=2^b$
			\begin{align}
				\mH_{g}^{(1)}(2^{a+1}&,2^b;z)-2b\mH_{g}(2^{a+1},2^b;z)=\sum_{\substack{a_0,{a_1},{b_0},{b_1}\geq0,\\a_0+{a_1}=a,{b_0}+{b_1}=b}}\sum_{\substack{g_0,g_1,k_1\geq0\\g_0+g_1+k_1=g}}\frac{1}{a+1}\nonumber\\
				& \times\frac{1}{(2k_1+1)!} \big(\mH_{g_0}^{(2)}(2^{a_0},2^{b_0};z)-2({a_0}+{b_0})\mH_{g_0}^{(1)}(2^{a_0},2^{b_0};z)+4{a_0}{b_0}       \nonumber\\
				& \times\mH_{g_0}(2^{a_0},2^{b_0};z) \big)\big(\mH_{g_1}^{(2k_1+2)}(2^{a_1},2^{b_1};z)-2{a_1}\mH_{g_1}^{(2k_1+1)}(2^{a_1},2^{b_1};z) \big),	\nonumber
			\end{align}
			and for $a=0,\,b=1$,
			\begin{align}
				\mH_{g}^{(1)}&(2,2;z)-2\mH_{g}(2,2;z)=\sum_{\substack{g_0,g_1,k_1\geq0\\g_0+g_1+k_1=g}}\frac{1}{(2k_1+1)!}\Big(\big(\mH_{g_0}^{(2)}(\varnothing,2;z)-2\mH_{g_0}^{(1)}(\varnothing,2;z)\big)\nonumber\\
				&  \times\mH_{g_1}^{(2k_1+2)}(\varnothing,\varnothing;z)+\mH_{g_0}^{(2)}(\varnothing,\varnothing;z)\big(\mH_{g_1}^{(2k_1+2)}(\varnothing,2;z)-2\mH_{g_1}^{(2k_1+1)}(\varnothing,2;z)\big)\Big).	\nonumber
			\end{align}
			$a=1,\,b=1$,
			\begin{align}
				\mH_{g}^{(1)}(2^2&,2;z)-2\mH_{g}(2^2,2;z)=\sum_{\substack{a_0,{a_1},{b_0},{b_1}\geq0,\\a_0+{a_1}=a,{b_0}+{b_1}=b}}\sum_{\substack{g_0,g_1,k_1\geq0\\g_0+g_1+k_1=g}}\frac{1}{2(2k_1+1)!}      \nonumber         \\
				& \times \Big(\big(\mH_{g_0}^{(2)}(2,2;z)-4\mH_{g_0}^{(1)}(2,2;z)+4\mH_{g_0}(2,2;z) \big)\mH_{g_1}^{(2k_1+2)}(\varnothing,\varnothing;z)\nonumber    \\
				&+2\big(\mH_{g_0}^{(2)}(2,\varnothing;z)-2\mH_{g_0}^{(1)}(2,\varnothing;z)\big)\big(\mH_{g_1}^{(2k_1+2)}(2,\varnothing;z)-\mH_{g_1}^{(2k_1+1)}(2,\varnothing;z)\big)\nonumber\\
				&+\mH_{g_0}^{(2)}(\varnothing,\varnothing;z)\big(\mH_{g_1}^{(2k_1+2)}(2,2;z)-2\mH_{g_1}^{(2k_1+1)}(2,2;z)\big)\Big),      	\nonumber
			\end{align}
		\end{corollary}
		
		Any element $f(T)\in\mathbb{Q}[T, \frac{1}{1-T}]$, admits a unique representation falling into one of the following three forms:
		\begin{align}
			&\text{(i) }f(T)=*\frac{1}{(1-T)^{j}}+*\frac{1}{(1-T)^{j-1}}+\cdots+*\frac{1}{(1-T)^{i+1}}+*\frac{1}{(1-T)^{i}}; \nonumber   \\
			&\text{(ii) }f(T)=*\frac{1}{(1-T)^{j}}+*\frac{1}{(1-T)^{j-1}}+\cdots+* T^{-i-1}+* T^{-i}; \nonumber   \\
			&\text{(iii) }f(T)=*T^{-j}+* T^{-j+1}+\cdots+* T^{-i-1}+* T^{-i}.\nonumber
		\end{align}
		In the above, the index $i$ is designated as the bottom bound of $f(T)$. For case (i) and (ii), the term involving the maximal exponent of $\frac{1}{1-T}$, is referred to as the top term of $f(T)$.	
		
		Although obtaining the value of $H^{(l)}_{g}(U,V ;T)$ by \eqref{D_H_r3} from $H_{g}(U,\varnothing ;T)$ is complex, it is well-suited for determining the top term and the bottom bound of $H_{g}(U,V ;T)$. Since 
		\begin{align}
			D\Big( \frac{1}{(1-T)^i}\Big)=\frac{i}{(1-T)^{i+2}}-\frac{i}{(1-T)^{i+1}},\quad i\neq0,  \label{DT}
		\end{align}
		the top term on the left hand of \eqref{D_H_r3} is that of $H^{(1)}_{g}(U,V ;T)$, while its bottom bound is that of $H_{g}(U,V ;T)$ or $H^{(1)}_{g}(U,V ;T)$ according as $V\neq\varnothing$ or $V=\varnothing$.
		
		\begin{lemma}\label{lr}
			The bottom bound of $H^{(\ell)}_{g}(U,V;T)$ is bounded below by
			\begin{align}
				2g-l^*(U)-\delta_{U,\varnothing}-l^*(V)-\delta_{V,\varnothing}+\ell-\delta_{{2g-l^*(U)-\delta_{U,\varnothing}-l^*(V)-\delta_{V,\varnothing}+\ell-1\leq0}}.\label{sbb}
			\end{align}
		\end{lemma}
		\begin{proof}
			Observe that the right-hand side of \eqref{D_H_r3} is the linear combination of the product of polynomials of the form 
			\begin{align}
				H_{g_0}^{(l_0)}(P^{(0)} ,\overline{P}^{(0)} ;T)\prod_{i=1}^{j}H_{g_i}^{(l_i+2k_i+\delta_{\text{odd}}^{l(R^{(i)})})}(P^{(i)*},\overline{P}^{(i)};T),\label{term}
			\end{align}
			where the following condition hold:
			\begin{align}
				&(i)\,\, Min\{l_0\}=\left\{
				\begin{aligned}
					&0\,\,&P^{(0)},\,\overline{P}^{(0)}\neq\varnothing\\
					&1      &\,\text{exactly one of }P^{(0)},\,\overline{P}^{(0)} \text{ is } \varnothing\\
					&2      &P^{(0)}=\overline{P}^{(0)}=\varnothing\\
				\end{aligned}
				\right.,\,
				Min\{l_i\}=\left\{
				\begin{aligned}
					&0\,\,&P^{(i)*}\neq\varnothing\\
					&1      &P^{(i)*}=\varnothing\\
				\end{aligned}
				\right.,\label{i}\\
				&(ii)\qquad\qquad\quad2g_0+\sum_{i=1}^{j}2g_{i}+2k_{i}=2g+2j-\sum_{i=1}^{j}(l(R^{(i)})+\delta_{\text{odd}}^{l(R^{(i)})}).\label{ii}
			\end{align}
			Step 1. In this step, we consider the case that one of $U$ and $V$ is the empty partition. Equation~\eqref{sym} implies that it suffices to consider the case that $V=\varnothing$. When $U=V=\varnothing$, \eqref{sbb} holds \cite{DYZ}. Assume that $\forall\, U=(u_1,\cdots)\models t,\,V=\varnothing$ with $u_1\leq n,\,\,m_n(U)\leq k
			,\,n\geq2,\,k\geq0$, \eqref{sbb} holds. We next prove \eqref{sbb} also holds for any $U$ with $u_1=n,\,m_n(U)=k+1$.
			
			Firstly, we consider the case $\ell=1$. By the inductive assumption together with \eqref{i} and ~\eqref{ii}, we deduce that the bottom bound of~\eqref{term} is bounded below by~\eqref{sbb}. Applying~\eqref{D_H_r3} with $V=\varnothing$ implies that the bottom bound of $H_g^{(1)}(U,\varnothing;T)$ is likewise bounded below by~\eqref{sbb}. 
			
			Secondly, we consider the case $\ell>1$, taking the derivative of $H^{(1)}_{g}(U,\varnothing ;T)$, gives the identical conclusion. 
			
			Thirdly, we consider the case $\ell=0$. When $2g-l^*(U)-1\leq1$, integrating $H^{(1)}_{g}(U,\varnothing ;T)$ yields the identical conclusion. When $2g-l^*(U)-1\geq2$, noticing that
			\begin{align}
				\frac{T^n}{(1-T)^m}=\left\{
				\begin{aligned}
					&\sum_{r=0}^{n}\binom{r}{n}(-1)^r\frac{1}{(1-T)^{m-r}},\,\,&n\leq m   \\
					&\substack{\sum_{r=0}^{m}\binom{r}{m}(-1)^r\frac{1}{(1-T)^{m-r}}\\+\sum_{k=0}^{n-m}(\sum_{r=k+m}^n\binom{r}{n}(-1)^{r+k})T^k}      &n>m\\
				\end{aligned}
				\right.,\label{numT}
			\end{align}
			by~\eqref{psiH2}, \eqref{psi7}, \eqref{psi5} and \eqref{psi6}, we have the bottom bound of $H_g(U,\varnothing;T)$ is likewise bounded below by $2g-l^*(U)-2$. Thus, integrating $H^{(1)}_{g}(U,\varnothing ;T)$, the $H_{g}(U,\varnothing ;T)$ have no constant term, which guarantee the bottom bound of $H_g(U,\varnothing;T)$ is likewise bounded below by~\eqref{sbb}. 
			
			Step 2. We assume that \eqref{sbb} holds for arbitrary partition $V$ and for all $U=(u_1,\cdots)\models t$ satisfying $u_1\leq n,\,\,m_n(U)\leq k$ where $n\geq2,\,k\geq0$. We proceed to prove that \eqref{sbb} then also holds for arbitrary partition $V$ and for any $U$ with $u_1=n,\,m_n(U)=k+1$.
			
			In the first case $\ell=0$, by the inductive assumption together with~\eqref{i} and ~\eqref{ii}, we deduce that the bottom bound of~\eqref{term}, satisfying the condition that $\overline{P}^{(0)}=V,\,\overline{P}^{(i)}=\varnothing$, is lower than that of any other term in~\eqref{term}, being equal to ~\eqref{sbb}.
			
			In the second case with $\ell>0$, the proof of \eqref{sbb} is completed by taking the derivative of $H_{g}(U,V ;T)$.
		\end{proof}

		\begin{lemma}\label{j}
			For all partitions $R^{(1)},\cdots,R^{(j)}$, define $K(R^{(1)},\,\cdots,\,R^{(j)}):=j-\frac{1}{2}\sum_{i=1}^{j}(l(R^{(i)})+\delta_{\text{odd}}^{l(R^{(i)})})$, then $K(R^{(1)},\,\cdots,\,R^{(j)})\leq0$. Moreover, $K(R^{(1)},\,\cdots,\,R^{(j)})=0$ if and only if $l(R^{(i)})=1$ or $2$.
			\begin{proof}
				For $1\leq k \leq j$, if $l(R^{(1)})=\cdots=l(R^{(k)})=2,\,l(R^{(k+1)})=\cdots=l(R^{(j)})=1$, then $K(R^{(1)},\,\cdots,\,R^{(j)})=0$. If $\exists R^{(i_0)}=(r_1,\,\cdots,\,r_m)$ with $m\geq3$, we have
				\begin{align}
					K(R^{(1)},\,\cdots,\,R^{(j)})=&j-\frac{1}{2}\sum_{i=1}^{j}(l(R^{(i)})+\delta_{\text{odd}}^{l(R^{(i)})})  \nonumber\\
					<&j+(m-1)-\frac{1}{2}(m+m)-\frac{1}{2}\sum_{i=1,\,i\neq i_0}^{j}(l(R^{(i)})+\delta_{\text{odd}}^{l(R^{(i)})})\nonumber\\
					=&K(R^{(1)},\,\cdots,\,\widehat{R^{(i_0)}},\,\cdots,\,R^{(j)},\,(r_1),\,\cdots,\,(r_m)).   
				\end{align}					
				If $R^{(1)},\,\cdots,\,\widehat{R^{(i_0)}},\,\allowbreak\cdots,\,R^{(j)},\,(r_1),\,\cdots,\,(r_m)$ still has a partition whose lengths are larger than 2, repeat the previous steps until the lengths of all partitions are less than or equal to 2. Then, we have $K(R^{(1)},\,\cdots,\,R^{(j)})<0$.
			\end{proof}
		\end{lemma}

		\begin{lemma}\label{pow_t_d}
			When $5g+2l(U)+2l(V)+2\ell\geq 6$, the top term of $H^{(\ell)}_{g}(U,V ;T)$ is given by
			\begin{align}
				\frac{(5g+2l(U)+2l(V)+2\ell-7)!!\cdot c_g}{24^g(5g-3)!! \#\mathrm{Aut}(U) \#\mathrm{Aut}(V)}\frac{\prod_{i}\frac{u_i^{u_i}}{u_i!} \prod_{j}\frac{v_j^{v_j}}{v_j!}}{(1-T)^{5g+2l(U)+2l(V)+2\ell-5}}.\nonumber
			\end{align}
			When $5g+2l(U)+2l(V)+2 \ell< 6$, $H_{g}^{(\ell)}(U,V ;T)\in \mathbb{Q}[T]$, unless $H_{1}(\varnothing,\varnothing ;T)=-\frac{T+\log(1-T)}{24}$.
			\begin{proof}
				When $U=V=\varnothing$, Lemma~\ref{pow_t_d} holds by \cite[Theorem~3]{DYZ}. Assume that $\forall\,U,\,V$ with $u_1\leq n,\,\,m_n(U)\leq k,\,v_1\leq m,\,\,m_m(V)\leq l ,\,n,\,m\geq2,\,k,\,l\geq0$, Lemma~\ref{pow_t_d} holds. Since $H_g^{(\ell)}(U,V;T)=H_g^{(\ell)}(V,U;T)$, which means we only need to prove Lemma~\ref{pow_t_d} also holds for any $U,V$ with $u_1=n$, $m_n(U)=k+1$, and $m_m(V)\leq l$.
				
				We denote the top term of $H^{(\ell)}_{g}(U,V ;T)$ with respect to $\frac{1}{1-T}$ as $\kappa ^\ell_{g}(U,V)$ where the relationship of $T,\,z$ is given in \eqref{z_T}. If $H^{(\ell)}_{g}(U,V ;T)\in \mathbb{Q}[T]$, $\kappa ^\ell_{g}(U,V)$ is defined as its constant term which equals $\mH_{g}^{(\ell)}(U,V;e^{-1})$ and its top power is defined to be zero. 
				
				In the first case $\ell=1$, let $A=U\setminus  u_1$. Referring back to \eqref{D_H_r3}, since $H_{1}(\varnothing,\varnothing;T)$ is absent from \eqref{D_H_r3}, we disregard its contribution. Notice that the right hand side of \eqref{D_H_r3} is the sum of product of polynomials $$H_{g_0}^{(l_0)}(P^{(0)} ,\overline{P}^{(0)} ;T)\prod_{i=1}^{j}H_{g_i}^{(l_i+2k_i+\delta_{\text{odd}}^{l(R^{(i)})})}(P^{(i)*},\overline{P}^{(i)};T).$$
				Their top powers are not all equal. We will therefore select the top terms of these polynomials whose power with respect to $\frac{1}{1-T}$ is the highest among others.
				
				By the inductive assumption and $2g_0+\sum_{i=1}^{j}2g_{i}+2k_{i}+l(R^{(i)})+\delta_{\text{odd}}^{l(R^{(i)})}=2g+2j$, the top powers of these polynomials with $k_i=0$ are larger than others.
				
				By Lemma~\ref{j}, when $5g+2l(U)+2l(V)+2\ell\geq 6$, the product of polynomials satisfying the following conditions possesses the largest top power among all such products:\\
				(a) $k_i=0,\,l(R^{(i)})\leq2,\,\forall 1\leq i\leq j$;\\
				(b) $\exists i_0,\,0\leq i_0\leq j$, $g_{i_0}=g,\,P^{(i_0)}=R^{(i_0)}\cup A,\,P^{(i_0)'}=V$;\\
				(c) $\forall i,\,i\neq i_0,\,g_{i}=0,\,P^{(i)}=R^{(i)},\,\overline{P}^{(i)}=\varnothing$.	\\			
				Their top powers are $5g+2l(A)+2l(V)-1$, i.e.,the top power of $H_g^{(1)}(U,V;T)$ are $5g+2l(A)+2l(V)-1$. Similarly, by Lemma~\ref{j}, when $g=0,\,A=V=\varnothing$, the top power of $H_0^{(1)}(u_1,\varnothing;T)$ is zero. 
				
				Then we investigate the coefficient of $(\frac{1}{1-T})^{5g+2l(U)+2l(V)+2\ell-5}$, when $5g+2l(U)+2l(V)+2\ell\geq 6$. From \cite{DYZ}, \eqref{psiH01Top}, \eqref{psiH02Top} and \eqref{psiH02Top2}, we obtain the following specific values: 
				\begin{align}
					&\kappa_0(\varnothing,\varnothing)=\frac{5}{12},\quad\kappa^1_0(\varnothing,\varnothing)=\frac{1}{2},\quad\kappa^2_0(\varnothing,\varnothing)=1,\quad\kappa_0(p,\varnothing)=\frac{p^{p -2}}{(p +1)!},\nonumber\\
					&\kappa^1_0(p,\varnothing)=\frac{p^{p -1}}{p!},\quad\kappa_0(p\,q,\varnothing)=\frac{p^{p}q^{q}}{p!q!(p+q)},\quad\kappa_0(p\,p,\varnothing)=\frac{p^{2p}}{4p\cdot p!^2}.\nonumber
				\end{align}
				Summing over the term of $(\frac{1}{1-T})^{5g+2l(U)+2l(V)+2\ell-5}$ in the right hand of \eqref{D_H_r3}, we derive the following recursion formula of the top term of $H_g^{(1)}(U,V;T)$.
				\begin{align}     				   
					&\kappa_{g}^{1}(U,V)= \sum_{j=1}^{u_1-1}\sum_{\substack{|R^{(1)}|+ \cdots  |R^{(j)}|=u_1-1\\l(R^{(1)}),\cdots,l(R^{(j)})\leq2}}\frac{2\mathcal{C}^{A}_{R^{(i_0)}\cup A}}{u_1j!m_{u_1}(U)m_1(R^{(i_0)})!} \kappa_{g}^{M(R^{(i_0)})}(R^{(i_0)*}\cup A,V)  \nonumber  \\
					&\times\prod_{\substack{i=1\\i\neq i_0}}^{j}\frac{2}{m_1(R^{(i)})!}\Big(\kappa_{0}^{M(R^{(i)})}(R^{(i)*},\varnothing)-\sum_{p\geq1}\frac{p^{p-1}}{(p+1)!}\delta^{R^{(i)}}_{(p,1)}\Big)+\sum_{j=1}^{u_1-1}\sum_{\substack{|R^{(1)}|+ \cdots  |R^{(j)}|=u_1-1\\l(R^{(1)}),\cdots,l(R^{(j)})\leq2}}  \nonumber   \\
					&\frac{\kappa_{g}^{2}(A,V)}{u_1j!m_{u_1}(U)} \prod_{i=1}^{j}\frac{2}{m_1(R^{(i)})!}\Big(\kappa_{0}^{M(R^{(i)})}(R^{(i)*},\varnothing)-\sum_{p\geq1}\frac{p^{p-1}}{(p+1)!}\delta^{R^{(i)}}_{(p,1)}\Big) ,     \label{toppower8}
				\end{align}
				where $M(R^{(i)})=m_1(R^{(i)})+\delta_{\text{odd}}^{l(R^{(i)})}$. We denote $R^{(i)*}$ in $\kappa_{0}^{M(R^{(i)})}(R^{(i)*},\varnothing)$ be the parts of $R^{(i)}$ which large or equal to $2$. Taking $g=0,\,A=V=\varnothing$ and $z=e^{-1}$ in \eqref{D_H_r3}, by Lemma~\ref{j}, we get
				\begin{align}
					\kappa^{1}_{0}(u_1,\varnothing)= \sum_{j=1}^{u_1-1}&\sum_{\substack{|R^{(1)}|+ \cdots  |R^{(j)}|=u_1-1\\l(R^{(1)}),\cdots,l(R^{(j)})\leq2}}\frac{1}{u_1j!} \prod_{i=1}^{j}\frac{2}{m_1(R^{(i)})!} \nonumber   \\	
					&\qquad\qquad\times\Big(\kappa_{0}^{M(R^{(i)})}(R^{(i)*},\varnothing)-\sum_{p\geq1}\frac{p^{p-1}}{(p+1)!}\delta^{R^{(i)}}_{(p,1)}\Big).   \label{toppower3}
				\end{align}		
				
				Let's define
				\begin{align}
					L_1(R^{(i_0)})=&\frac{2\mathcal{C}^{A}_{R^{(i_0)}\cup A}}{m_1(R^{(i_0)})!}\kappa_{g}^{M(R^{(i_0)})}(R^{(i_0)*}\cup A,V),\nonumber\\
					L_2(R^{(i_0)})=&\frac{2}{m_1(R^{(i_0)})!}\Big(\kappa_{0}^{M(R^{(i_0)})}(R^{(i_0)*},\varnothing)-\sum_{p\geq1}\frac{p^{p-1}}{(p+1)!}\delta^{R^{(i_0)}}_{(p,1)}\Big).\nonumber
				\end{align}
				Note that the relationship $$|L_1(R^{(i_0)})|=|R^{(i_0)}|\kappa_g^2(A,V)|L_2(R^{(i_0)})|,$$holds for all considered cases of $R^{(i_0)}$, namely $(1),\,(1,1),\,(p,1),\,(p,p),\,(p,q)$ with $p\neq q\geq2$. By \eqref{toppower3}, we have
				\begin{align}     				   
					\kappa_{g}^{1}(U,V)=&\frac{1}{m_{u_1}(U)}\kappa^{1}_{0}(u_1,\varnothing)\kappa_g^2(A,V) +\frac{u_1-1}{m_{u_1}(U)}\kappa^{1}_{0}(u_1,\varnothing)\kappa_g^2(A,V)\nonumber\\
					=&\frac{(5g+2l(U)+2l(V)-5)!!c_g}{24^g(5g-3)!! \#\mathrm{Aut}(U) \#\mathrm{Aut}(V)}\frac{\prod_{i}\frac{u_i^{u_i}}{u_i!} \prod_{j}\frac{v_j^{v_j}}{v_j!}}{(1-T)^{5g+2l(U)+2l(V)-3}}.\label{toppower10}
				\end{align}
				
				In the second case $\ell=0$, by $D=\frac{T}{1-T}\frac{d~}{dT}$, $H_{g}(U,V;T)$ satisfy Lemma~\ref{pow_t_d}, for $5g+2l(U)+2l(V)\geq 6$. When $g=0,\,U=n,\,V=m$, taking $g=0,\,U=u_1,\,V=v_1$ in \eqref{D_H_r}, by the inductive assumption, we have $H_{0}(u_1,v_1;T)$ satisfy Lemma~\ref{pow_t_d}.		
				
				In the remaining case $\ell\geq2$, differentiating $\kappa^{1}_{g}(U,V)$ shows that $\kappa_{g}^{\ell}(U,V)$ satisfies Lemma~\ref{pow_t_d}. The only potential exceptions are $\kappa^{\ell}_0(u_1,\varnothing)$ with $\ell>1,u_1\geq2$, for which $\kappa^1_0(u_1,\varnothing)=\frac{u_1^{u_1 -1}}{u_1!}$ is a constant; however, \eqref{Pand} and \eqref{psiH01Top} imply that these terms also satisfy Lemma~\ref{pow_t_d}.		
			\end{proof}
		\end{lemma}		
		
		\begin{proof}[\bf{Proof of Theorem~\ref{thm_t_d}}]
			For $U\models t,\,V\models w,\,t,w\in \mathbb{N},$ let $\ell=0$ in Lemma~\ref{lr}. We have the generating series of double Hurwitz numbers with a fixed genus $H_{g}(U,V;T)$ belong to $\mathbb{Q}[\frac{1}{1-T}]$, when $2g\geq  l^*(U)+{l}^*(V)+\delta_{U,\varnothing}+\delta_{V,\varnothing}+1$. In this case, Lemma~\ref{pow_t_d} gives its top term.
		\end{proof}
		Let us recall the definition of the Lambert module. The ring $\mathbb{Q}[\frac{1}{1-T}]$ is a free $\mathbb{Q}[D]$-module of rank 2 called the Lambert module:
		\begin{align}
			\Lambda:=\mathbb{Q}[D] \overline{\alpha}\oplus\mathbb{Q}[D]\overline{\beta},
		\end{align}		
		with $\overline{\alpha}=\frac{T}{1-T},\, \overline{\beta}=\frac{T}{(1-T)^2}$ (for details, see \cite{DYZ}). Base on the above result, we obtain the following Theorem.
		\begin{theorem}\label{h}
			For fixed $U\models t,\,V\models w,\,t,w\in \mathbb{N}$ and $g$ with $2g\geq {l}^*( U)+{l}^*( V)+\delta_{U,\varnothing}+\delta_{V,\varnothing}-2$, 
			\begin{align}
				&h_{g,d}( U\,1^{d-t}, V\,1^{d-w},2\,1^{d-2},2\,1^{d-2},\cdots)= f_0h_{0,d}+f_1h_{1,d},
			\end{align}
			where $f_0=f_0( U, V,g;d),\,f_1=f_1( U, V,g;d)\in \mathbb{Q}[d]$ with $deg(f_0)=[\frac{5}{2}g]+l({U})+l(V),\,deg(f_1)=[\frac{5g-5}{2}]+l({U})+l(V)$. 
			
			Furthermore, for fixed $U\models t,\,V\models w$ and $g\geq0$,
			\begin{align}
				&h_{g,d}( U\,1^{d-t}, V\,1^{d-w},2\,1^{d-2},2\,1^{d-2},\cdots)= \frac{f_0}{d^n}h_{0,d}+f_1h_{1,d},
			\end{align}
			where $n={l}^*( U)+{l}^*( V)+\delta_{U,\varnothing}+\delta_{V,\varnothing}-2g-2$ and $deg(f_0)=[\frac{5}{2}g]+l({U})+l(V)+n$.
			
			\begin{proof}
				Since $H^{(3)}_{0}(\varnothing,\varnothing;T)=\frac{1}{1-T}-1$ and $24H_{1}^{(1)}(\varnothing,\varnothing;T)+H^{(3)}_{0}(\varnothing,\varnothing;T)=\frac{1}{(1-T)^2}-\frac{1}{1-T}$ \cite{DYZ}, we have $\frac{1}{(1-T)^{i-1}}-\frac{1}{(1-T)^{i}}$ is a linear combination of $H^{(3)}_{0}(\varnothing,\varnothing;T),\,\cdots,\allowbreak H^{(3+[\frac{i-1}{2}])}_{0}(\varnothing,\varnothing;T),\,H_{1}^{(1)}(\varnothing,\varnothing;T),\allowbreak\cdots,\,H_{1}^{(1+[\frac{i-2}{2}])}(\varnothing,\varnothing;T)$. By Lemma~\ref{lr}, we have that $H_g(U,V;T)$ belong to $\mathbb{Q}[\frac{1}{1-T}]\oplus \text{span}\{1,T,T^2,T^3\}$ with rational coefficients, for $2g\geq {l}^*(U)+{l}^*(V)+\delta_{U,\varnothing}+\delta_{V,\varnothing}-2$. Since $H_{0}(\varnothing,\varnothing;T)=T-\frac{3}{4}T^2+\frac{1}{6}T^3$, by the definition of $H_g(U,V;T)$ \eqref{Hg2}, the result follows.
			\end{proof}
		\end{theorem}
		
		\begin{proof}[\bf{Proof of Corollary~\ref{rmk}}]
			Without loss of generality, we assume $5g+2l(U)+2l(V)\geq6$. When $5g+2l(U)+2l(V)<6$, by Lemma~\ref{pow_t_d}, we use $\mH_{g}^{(\ell)}(U,V;z)=\mH_{g}^{(\ell)}(U,V;z(T)):=H_{g}^{(\ell)}(U,V;T)$ instead of $\mH_{g}(U,V;z)$ with an appropriate $\ell$ such that $5g+2l(U)+2l(V)+2\ell\geq 6$. Since \eqref{z_T}, we have $1-ez=\frac{1}{2}(1-T)^2+O((1-T)^3)$. Letting $T$ and $z$ in \eqref{conj}, \eqref{Hg2} tend to $1$ and $e^{-1}$ respectively, we have
			\begin{align}
				&\sum\limits_{d} z^d h_{g,d}(U\,1^{d-t},\,V\,1^{d-w},2\,1^{d-2},\cdots)\sim\frac{(5g+2l(U)+2l(V)-7)!!\cdot c_g}{24^g(5g-3)!! \#\mathrm{Aut}(U) \#\mathrm{Aut}(V)} \nonumber\\
				&\qquad\qquad\quad\,\,\times\prod_{i=1}^{l(U)}\frac{u_i^{u_i}}{u_i!} \prod_{j=1}^{l(V)}\frac{v_j^{v_j}}{v_j!}\frac{1}{(2(1-ez))^{\frac{5g+2l(U)+2l(V)-5}{2}}}   , \quad \text{as }z\rightarrow e^{-1}. \label{rmk2}
			\end{align}
			Notice that when $n$ tends to $\infty$, the coefficient of $x^n$ in the expansion of $(1-x)^{-\alpha}$ is $\frac{n^{\alpha-1}}{\Gamma(\alpha)},\,\forall\alpha\in\mathbb{C}\setminus  \mathbb{Z}_{\leq0}$ at the singularity $x=1$, as shown in Theorem~VI.1 of \cite{FS}. Letting $d$ tends to $\infty$ and comparing the coefficients of $z^d$ on both sides of \eqref{rmk2}, we have
			\begin{align}
				h_{g,d}(U\,1^{d-t},\,V\,1^{d-w},2\,1^{d-2},\cdots)\sim\frac{\prod_{i}\frac{u_i^{u_i}}{u_i!} \prod_{j}\frac{v_j^{v_j}}{v_j!}\sqrt{2}e^{d} d^{\frac{5g-7}{2}+l(U)+l(V)} \cdot c_g}{(96\sqrt{2})^g \Gamma(\frac{5g-1}{2}) \#\mathrm{Aut}(U) \#\mathrm{Aut}(V)}  \nonumber  
			\end{align}
			Using Stirling's approximation, we obtain
			\begin{align}
				&H_{g,d}(U\,1^{d-t},\,V\,1^{d-w},2\,1^{d-2},\cdots)\sim\frac{\prod\frac{u_i^{u_i}}{u_i!} \prod\frac{v_j^{v_j}}{v_j!}2^{-l^*(U)-l^*(V)} \sqrt{\pi/2}\cdot c_g}{(24\sqrt{2})^g \Gamma(\frac{5g-1}{2}) \#\mathrm{Aut}(U) \#\mathrm{Aut}(V)}\left(\frac{4}{e}\right)^{d}  \nonumber  \\
				&  \qquad\qquad\qquad\qquad\qquad\qquad\quad   \times  d^{2d-5+9g/2+2l(U)+2l(V)-t-w},\quad \,\,\,    \text{as }d\rightarrow \infty.\nonumber
			\end{align}
		\end{proof}
		Based on the numerical experiments, we propose the following conjecture.
		\begin{conjecture}\label{1}
			For any $s\geq 0$ and $5g+2l(\mu^{(1)})+\dots+2l(\mu^{(s)})\geq 6$, the top term of $H_{g}(\mu^{(1)},\dots,\mu^{(s)} ;T)$ is given by
			\begin{align}
				\frac{(5g+\sum_{i=1}^s 2l(\mu^{(i)})-7)!!\cdot c_g}{24^g(5g-3)!! \prod_{i=1}^s \#\mathrm{Aut}(\mu^{(i)})} \frac{\prod_{i=1}^s\prod_{j}\frac{(\mu^{(i)}_j)^{\mu^{(i)}_j}}{\mu^{(i)}_j!}}{(1-T)^{5g+\sum_{i=1}^s 2l(\mu^{(i)})-5}}.\nonumber
			\end{align}
		\end{conjecture}
		\begin{proposition}\label{2}
			Assume that Conjecture~\ref{1} holds. Then for any fixed $s,g\geq 0$, the asymptotics of $H_{g,d}(\mu^{(1)}\,1^{d-|\mu^{(1)}|},\allowbreak\dots,\mu^{(s)}\,1^{d-|\mu^{(s)}|},2\,1^{d-2},2\,1^{d-2},\cdots)$ is given by
			\begin{align}
				&H_{g,d}(\mu^{(1)}\,1^{d-|\mu^{(1)}|},\dots,\mu^{(s)}\,1^{d-|\mu^{(s)}|},2\,1^{d-2},\cdots)\sim\frac{\prod_{i=1}^s\prod_{j}\frac{(\mu^{(i)}_j)^{\mu^{(i)}_j}}{\mu^{(i)}_j!} 2^{-\sum_{i=1}^s l^*(\mu^{(i)})} \sqrt{\pi/2}\cdot c_g}{(24\sqrt{2})^g \Gamma(\frac{5g-1}{2}) \prod_{i=1}^s \#\mathrm{Aut}(\mu^{(i)})}  \nonumber  \\
				&  \qquad\qquad\qquad\quad   \times  \left(\frac{4}{e}\right)^{d}d^{2d-5+9g/2+\sum_{i=1}^s 2l(\mu^{(i)})-|\mu^{(i)}|},\quad \,\,\,    \text{as }d\rightarrow \infty.\label{bh}
			\end{align}
		\end{proposition}

		\section{Structure with a fixed degree}\label{CC}
		In this section, we study the recursions and structures for the generating series of double Hurwitz numbers with a fixed degree which implies the asymptotics of double Hurwitz numbers.
		\begin{proposition}
			For $\mu^{(1)}:=(a)\cup A=(a,\cdots)\vdash d,\,\mu^{(2)}:=(b)\cup B={(b,\cdots)}\vdash d,\,a,\,b\geq2,d\in\mathbb{N}$, we have \eqref{D_C1} and
			\begin{align}      
				C_d&(\mu^{(1)},\mu^{(2)};x)=\sum_{r=1}^{b-1}\sum_{i_0=1}^{r+1}\sum_{\substack{P\in\operatorname{Split}_{i_0}^{r+\delta^{i_0}_{r+1}}[U]\\(\overline{P},S)\in\operatorname{Part}^{r+\delta^{i_0}_{r+1}}[V]\\|P^{(i)}|=|\overline{P}^{(i)}|=d_i \\0\leq i\leq r+\delta^{i_0}_{r+1}  }}\frac{(a-1)(1+m_{1}( P^{(0)}))}{2br!m_{a}(\mu^{(1)})m_{b}(\mu^{(2)})}C_{d_0+1}(1\cup P^{(0)},1\cup \overline{P}^{(0)};x)\nonumber       \\			
				&  \times  (1+m_{1}(\overline{P}^{(0)}))m_{a-1}(P^{(i_0)})\prod_{i=1}^{r+\delta^{i_0}_{r+1}} \mathcal{C}^{\overline{P}^{(i)}}_{S^{(i)}}\varepsilon_{{l(S^{(i)})+\delta^i_{i_0}}}(d_i x)C_{d_i}( P^{(i)},\overline{P}^{(i)};x) -\sum_{w=1}^{[\frac{a}{2}]+1} \sum_{j_0=1}^{w}  \nonumber\\
				&   \times\sum_{\substack{(P,R)\in\overline{\operatorname{Part}}^{w}[U]\\\overline{P}\in\overline{\operatorname{Split}}_{j_0}^{w}[V]\\\\|P^{(j)}|=|\overline{P}^{(j)}|=d_j \\0\leq i\leq w}} \frac{2^{w}m_b(\overline{P}^{(j_0)})}{2w!m_{a}(\mu^{(1)})m_{b}(\mu^{(2)})} \prod_{j=1}^{w}\mathcal{C}^{ P^{(j)}}_{R^{(j)}}C_{d_j}( P^{(j)},\overline{P}^{(j)};x).        \label{D_C2} 
			\end{align}	
			\begin{proof}
				Let $\mu^{(1)}=U \cup 1^{d-|U|},\,\mu^{(2)}=V\cup 1^{d-|V|}$. Comparing coefficients of $y^d$ on both sides of ({\ref{Pand4}}) and \eqref{D_Pand}, we have \eqref{D_C1} and \eqref{D_C2}.
			\end{proof}
		\end{proposition}
		\begin{example}
			When $a=2$, let partitions $A=(2^{u},\,1^v)$, then \eqref{D_C1} becomes 
			\begin{align}
				& C_d(2^{u+1}\,1^v,1^{2u+v+2};x)= \sum_{\substack{u_1+u_2=u\\v_1+v_2=v}}\frac{v_1(2u_1+v_1)(v_2+1)}{2(u+1)(2u+v+2)}C_{2u_1+v_1}(2^{u_1}\,1^{v_1},1^{2u_1+v_1};x)   \nonumber \\ 
				&\qquad\quad\big(e^{(2u_2+v_2+1)x}-e^{-(2u_2+v_2+1)x}\big)C_{2u_2+v_2+1}(2^{u_2}\,1^{v_2+1},1^{2u_2+v_2+1};x).     \label{D_C3}
			\end{align}
			For $u=0,\,v=0$ or $u=0,\,v=1$,
			\begin{align}
				C_2(2,1^2;x)= &\frac{1}{4}C_{1}(1,1;x)(e^{x}-e^{-x})C_{1}(1,1;x),     \nonumber\\
				C_3(2\,1,1^3;x)= &\frac{1}{3}C_{1}(1,1;x)(e^{2x}-e^{-2x})C_{2}(1^2,1^2;x)\nonumber\\
				&+ \frac{2}{3}C_{2}(1^2,1^2;x)(e^{x}-e^{-x})C_{1}(1,1;x).     \nonumber
			\end{align}
			Let $a=3,\,b=2,\,A=\varnothing,\,B=(1)$, then \eqref{D_C2} becomes 
			\begin{align}
				C_2(2,1^2;x)= &0+C_{1}(1,1;x)(e^{2x}-e^{-2x})C_{2}(2,1^2;x)+C_{3}(2\,1,1^3;x).
			\end{align}
		\end{example}
		The following elementary identities play an important role in the proof later:
		\begin{align}
			&\sinh(k_1 x)\cdots \sinh(k_i x)\cosh(k_{i+1}x)\cdots \cosh(k_{i+j}x)\nonumber\\
			&\,\,=\left\{
			\begin{aligned}
				&\frac{1}{2^{i+j}}\sinh(k_1+\cdots+k_{i+j})x+\sum_{l=1}^{k_1+\cdots+k_{i+j}-1}a_l \sinh(l x)\,\,&i\text{ is odd}\\
				&\frac{1}{2^{i+j}}\cosh(k_1+\cdots+k_{i+j})x+\sum_{l=0}^{k_1+\cdots+k_{i+j}-1}a_l \cosh(l x)      &i\text{ is even}\\
			\end{aligned}
			\right..\label{shch1}
		\end{align}
		\begin{lemma}\label{eh}
			For all $n\in\mathbb{N}$,
			\begin{align}
				\sum_{\lambda,\lambda\vdash n}\frac{1}{z_\lambda}=1.
			\end{align}
			When $n\geq2$,
			\begin{align}
				\sum_{\lambda,\lambda\vdash n,l(\lambda)=\text{odd}}\frac{1}{z_\lambda}=\frac{1}{2}.
			\end{align}
			\begin{proof}Notice that
				\begin{align}
					h_n(\bx)=\sum_{\lambda,\lambda\vdash n}\frac{1}{z_\lambda}p_\lambda(\bx), \qquad e_n(\bx)=&\sum_{\lambda,\lambda\vdash n}\frac{(-1)^{n-l^{(\lambda)}}}{z_\lambda}p_\lambda(\bx),\label{z2}
				\end{align}
				with $t_j=\frac{1}{j}\sum_i x_i^j$ as $\bt$ in \eqref{rlt}, see for example Section 1.2 in \cite{Mac}. When the independent variable is written as $\bx$ rather than $\bt$,
				\begin{align}
					h_n(\bx)=&\sum_{i_1\leq i_2\leq \cdots\leq i_n}x_{i_1}x_{i_2}\cdots x_{i_n}, \quad p_k(\bx)=\sum_i x_i^k,   \nonumber\\
					e_n(\bx)=&\sum_{i_1< i_2< \cdots< i_n}x_{i_1}x_{i_2}\cdots x_{i_n}.\nonumber
				\end{align}
				Let $x_1=1,\,x_2=x_3=\cdots=0$ in \eqref{z2}, the result follows.
			\end{proof}
		\end{lemma}
		
		By using a method of Dubrovin-Yang-Zagier \cite{DYZ}, we give the proof of Theorem~\ref{tr2} in the following, which generalizes part of \cite[Theorem~1]{DYZ}.
		
		\begin{proof}[\bf{Proof of Theorem~\ref{tr2}}] Setp 1. We prove that for all $\mu\vdash d$ and $\nu=1^d$, Theorem~\ref{tr2} is true. When $\mu=1^d$, this statement was proved in \cite[Theorem~1]{DYZ}. Assume that $\forall \mu$ with $\mu_1\leq n,\,m_n(\mu)\leq k-1$, where $n\geq2,\,k\geq1$, $\nu=1^{|\mu|}$, this statement holds. Then, for $\mu$ with $\mu_1\leq n,\,m_n(\mu)= k$, referring back to \eqref{D_C1}, let $\nu=1^d$. By the inductive assumption and \eqref{shch1}, we have
			\begin{align}
				C_d(\mu,1^d;x)=\left\{
				\begin{aligned}
					&\sum_{k=0}^{\infty}\kappa(\mu,1^d;k)\cosh(kx),\,\,&d+l(\mu)=\text{even}\\
					&\sum_{k=0}^{\infty}\kappa(\mu,1^d;k)\sinh(kx)      &d+l(\mu)=\text{odd}\\
				\end{aligned}
				\right.,
			\end{align}
			and for any partition $\sigma$ with $m_n(\sigma)\leq k-1$, the largest $k$ in $cosh(kx)$ or $sinh(kx)$ in $C_d(\sigma,1^{|\sigma|};x)$ is $\binom{|\sigma|}{2}$ with coefficient
			\begin{align}
				\kappa(\sigma,1^{|\sigma|};\tbinom{|\sigma|}{2})=\frac{2}{z_\sigma |\sigma|!}.       \label{k}
			\end{align}
			
			Notice that for $\sum_{i=0}^j d_i=d-1$,
			\begin{align}
				\binom{d_0+1}{2}+\sum_{i=1}^j d_i+\binom{d_i}{2}\leq 0+d-1+\binom{d-1}{2}=\binom{d}{2}.\nonumber
			\end{align}
			By \eqref{shch1}, the largest $k$ in $cosh(kx)$ or $sinh(kx)$ in $C_d(\mu,1^d;x)$ is $\tbinom{d}{2}$. Furthermore, the recursion formulas of the top coefficient of $C_d(\mu,1^d;x)$ is
			\begin{align}		 
				\kappa(\mu,1^d;\tbinom{d}{2})= \sum_{\lambda\vdash n-1}& \frac{\mathcal{C}^{ \lambda \cup A}_{\lambda }}{nd! m_n(\mu)}\kappa(\lambda\cup A,1^{d-1};\frac{(d-1)(d-2)}{2}),
			\end{align}		
			where $A=\mu\setminus n$,. By \eqref{k} and Lemma~\ref{eh}, we obtain
			\begin{align}		 
				& \kappa(\mu,1^{d};\tbinom{d}{2})=  \frac{2}{z_\mu z_{1^d}}\times\sum_{\lambda\vdash n-1}\frac{1}{z_\lambda}=\frac{2}{z_\mu z_{1^d}}.
			\end{align}		
			
			Step 2. We prove that for all $\mu,\,\nu,\,|\mu|=|\nu|$, Theorem~\ref{tr2} is true. When $\nu=1^d$, this statement was proved before. Assume that $\forall \mu,\nu,\,|\mu|=|\nu|$ with $\mu_1\leq a,\,\,m_{a}(\mu)\leq k-1,\,\nu_1\leq b,\,\,m_{b}(\nu)\leq l-1 ,\,a,\,b\geq2,\,k,\,l\geq1$, this statement holds. Then for $\mu,\nu,\,|\mu|=|\nu|$ with $\mu_1=a,\,\nu_1=b,\,m_{a}(\mu)=k,\,m_{b}(\nu)=l$, referring back to \eqref{D_C2}, let $A=\mu\setminus a,\,B=\nu\setminus b$. By the inductive assumption and \eqref{shch1}, we have
			\begin{align}
				C_d(\mu,\nu;x)=\left\{
				\begin{aligned}
					&\sum_{k=0}^{\infty}\kappa(\mu,\nu;k)\cosh(kx),\,\,&l(\mu)+l(\nu)=\text{even}\\
					&\sum_{k=0}^{\infty}\kappa(\mu,\nu;k)\sinh(kx)      &l(\mu)+l(\nu)=\text{odd}\\
				\end{aligned}
				\right..
			\end{align}
			and for partitions $\sigma,\,\omega,\,|\sigma|=|\omega|$ with $m_a(\sigma)\leq k-1,\,m_b(\sigma)\leq l-1$, the largest $k$ in $cosh(kx)$ or $sinh(kx)$ in $C_d(\sigma,\omega;x)$ is $\binom{|\sigma|}{2}$ with coefficient
			\begin{align}
				\kappa(\sigma,\omega;\tbinom{|\sigma|}{2})=\frac{2}{ z_\sigma z_\omega}.       \label{k2}
			\end{align}
			Then we divide the right hand side of \eqref{D_C2} into two parts to discuss by summation sign. In first part, for $\sum_{i=0}^{r+\delta^{i_0}_{r+1}} d_i=d-1,\,1\leq i_0\leq r+1$, notice that
			\begin{align}
				\binom{d_0+1}{2}+\sum_{i=1}^{r+\delta^{i_0}_{r+1}} d_i+\binom{d_i}{2}< 0+d-1+\binom{d-1}{2}=\binom{d}{2}. \nonumber
			\end{align}
			By \eqref{shch1}, the largest $k$ in $cosh(kx)$ or $sinh(kx)$ in first part is $\tbinom{d}{2}$.
			
			In second part, when $a\geq3$, notice that $d_j\geq 0,\,\sum_{j=1}^w d_j=d$. By \eqref{shch1}, the largest $k$ in $cosh(kx)$ or $sinh(kx)$ in second part is $\tbinom{d}{2}$.
			
			So the largest $k$ in $cosh(kx)$ or $sinh(kx)$ in $C_d(\mu,\nu;x)$ is $\tbinom{d}{2}$. Furthermore, when $a=2$, the recursion formulas of the top coefficient of $C_d(\mu,\nu;x)$ is
			\begin{align}		 
				\kappa(\mu,\nu;\tbinom{d}{2})=&\sum_{| \lambda|=b-1} \frac{ m_{ 1}((1)\cup A )\mathcal{C}^{ \lambda \cup A}_{\lambda }}{2b\cdot m_2(\mu) m_b(\nu)}\kappa((1)\cup A , \lambda\cup B ;\frac{(d-1)(d-2)}{2}).
			\end{align}		
			By \eqref{k2} and Lemma~\ref{eh},
			\begin{align}		 
				& \kappa(\mu,\nu;\tbinom{d}{2})=\frac{2}{z_\mu z_{\nu}}\times\sum_{\lambda\vdash m-1}\frac{1}{z_\lambda}=\frac{2}{z_\mu z_{\nu}}.
			\end{align}		
			When $a\geq3$, the recursion formulas of the top coefficient of $C_d(\mu,\nu;x)$ is
			\begin{align}      
				\kappa(\mu,\nu;x)=&\sum_{| \lambda|=b-1} \frac{ m_{a-1}((a-1)\cup A)}{2b m_a(\mu) m_b(\nu)}\kappa((a-1)\cup A, \lambda\cup B;\frac{(d-1)(d-2)}{2}) \nonumber\\			
				&\times  (a-1)\mathcal{C}^{ \lambda \cup B}_{\lambda }-\sum_{\substack{\lambda\vdash a, \lambda_1\neq a\\l(\lambda)=odd}}\frac{ 1}{m_a(\mu) }\mathcal{C}^{ \lambda \cup A}_{\lambda }\kappa( \lambda\cup A,\nu;x).  
			\end{align}
			By \eqref{k2} and Lemma~\ref{eh}
			\begin{align}		 
				\kappa(\mu,\nu;x)=&\frac{a}{z_\mu z_{\nu}}\Big(\sum_{\lambda\vdash b-1}\frac{1}{z_\lambda}\Big)-\frac{2a}{z_\mu z_{\nu}}\Big(\sum_{\substack{\lambda\vdash a, \lambda_1\neq a\\l(\lambda)=odd}}\frac{1}{z_\lambda}\Big)=\frac{2}{z_\mu z_{\nu}}.
			\end{align}					
		\end{proof}
		
		\begin{corollary}\label{cr}
			The generating series of not-necessarily connected double Hurwitz numbers with a fixed degree
			\begin{align}
				C_d^*(\mu^{(1)},\mu^{(2)};x):=\sum_{k\geq0} \frac{1}{k!}H_{k,d}^*(\mu^{(1)},\mu^{(2)})x^{k}
			\end{align}
			have structure 
			\begin{align}
				C^*_d(\mu^{(1)},\mu^{(2)};x)=\left\{
				\begin{aligned}
					&\sum_{k=0}^{\frac{d(d-1)}{2}}\kappa^*(\mu^{(1)},\mu^{(2)};k)cosh(kx),\,&l(\mu^{(1)})+l(\mu^{(2)})=\text{even}\\
					&\sum_{k=0}^{\frac{d(d-1)}{2}}\kappa^*(\mu^{(1)},\mu^{(2)};k)sinh(kx),      &l(\mu^{(1)})+l(\mu^{(2)})=\text{odd}\\
				\end{aligned}
				\right.\nonumber
			\end{align}
			and 			
			\begin{align}
				\kappa^*(\mu^{(1)},\mu^{(2)};\tbinom{d}{2})=\frac{2}{z_{\mu^{(1)}} z_{\mu^{(2)}}}.
			\end{align}
			Moreover, 
			\begin{align}
				H_{k,d}^*(\mu^{(1)},\mu^{(2)}) \sim&\frac{2}{z_{\mu^{(1)}} z_{\mu^{(2)}}}\binom{d}{2}^{k},\qquad\quad\,\,\text{as }k\rightarrow \infty.\nonumber
			\end{align}
		\end{corollary}
		For the case when $\mu^{(1)}=\mu^{(2)}=1^d$, Corollary~\ref{cr} was given in \cite{DYZ}. 
		\begin{proof}
			In fact, the partition function of connected double Hurwitz numbers is the generating series of not-necessarily connected double Hurwitz numbers (cf. \cite[Chapter 3]{GJ}), namely,
			\begin{align}
				\mZ(\bp,\bp';x,y)=\sum_{k,d} \sum_{\substack{\mu^{(1)},\mu^{(2)}\vdash d}}\frac{1}{k!}H^*_{k,d}(\mu^{(1)},\mu^{(2)}) x^ky^d .\label{mZ}
			\end{align}		
			By Theorem~\ref{rmk} and \eqref{tr2}, the result follows.
		\end{proof}

		\section*{Acknowledgement}
		I would like to thank Di Yang for his advice and suggestion of questions. I would like to thank Xiaofeng Chen and Chenglang Yang for careful read of an early version of the manuscript. This work was supported by the NSFC No. 12371254 and CAS No. YSBR-032.
		
		\appendix
		
		\section{Examples of generating series for double Hurwitz numbers}
		\begin{center}\label{Hur}
			\renewcommand\arraystretch{2}
			\tabcolsep=0.1cm
			\begin{longtable}{|p{1.8cm}|p{3.65cm}|p{5.95cm}|} 
				\hline
				\tiny
				\diagbox{\,\,\,\,$({\mu},{\nu})$\,}{$g$\quad}&\centering{2} &\qquad\qquad\qquad\qquad\qquad3 \\ \hline
				$\mH_{g}(\varnothing,\varnothing;z)$&$-\frac{1}{240 (1-T)^2}+\frac{19}{1440 (1-T)^3}$ $-\frac{1}{72 (1-T)^4}+\frac{7}{1440 (1-T)^5} $&$ \frac{1}{1008 (1-T)^4}-\frac{113}{10080 (1-T)^5}+\frac{2383}{51840 (1-T)^6}$ $-\frac{16759}{181440 (1-T)^7}+\frac{227}{2304 (1-T)^8}-\frac{557}{10368 (1-T)^9}$ $+\frac{245}{20736 (1-T)^{10}} $\\ \hline
				$\mH_{g}(2^2,\varnothing;z)$&$-\frac{1}{90} +\frac{1}{90 (1-T)}$ $+\frac{29}{240 (1-T)^2}-\frac{253}{720 (1-T)^3}$ $+\frac{307}{720 (1-T)^4}-\frac{5}{144 (1-T)^5}$ $-\frac{731}{720 (1-T)^6}+\frac{29}{16 (1-T)^7}$ $-\frac{187}{144 (1-T)^8}+\frac{49}{144 (1-T)^9}$&$ \frac{13}{2520 (1-T)^3}-\frac{107}{1120 (1-T)^4}+\frac{2759}{4320 (1-T)^5}$ $-\frac{183769}{90720 (1-T)^6}+\frac{129737}{45360 (1-T)^7}+\frac{18509}{10080 (1-T)^8}$ $-\frac{114041}{6480 (1-T)^9}+\frac{104213}{2592 (1-T)^{10}}-\frac{36907}{720 (1-T)^{11}}$ $+\frac{50117}{1296 (1-T)^{12}}-\frac{41659}{2592 (1-T)^{13}}+\frac{1225}{432 (1-T)^{14}} $\\ \hline
				$\mH_{g}(2^2,2^2;z)$&$\frac{T}{9}+\frac{23}{90}-\frac{19}{45 (1-T)}$ $+\frac{43}{48 (1-T)^2}+\frac{1081}{720 (1-T)^3}$ $-\frac{8899}{720 (1-T)^4}+\frac{9553}{720 (1-T)^5}$ $+\frac{4345}{144 (1-T)^6}-\frac{26653}{240 (1-T)^7}$ $+\frac{14125}{144 (1-T)^8}+\frac{28369}{144 (1-T)^9}$ $-\frac{218963}{360 (1-T)^{10}}+\frac{47513}{72 (1-T)^{11}}$ $-\frac{24221}{72 (1-T)^{12}}+\frac{539}{8 (1-T)^{13}} $&$ \frac{53}{3780 (1-T)^2}-\frac{37}{360 (1-T)^3}+\frac{6721}{3360 (1-T)^4}$ $-\frac{1390379}{90720 (1-T)^5}+\frac{1026953}{18144 (1-T)^6}-\frac{3889903}{45360 (1-T)^7}$ $-\frac{2222317}{12960 (1-T)^8}+\frac{6643993}{5670 (1-T)^9}-\frac{71779133}{30240 (1-T)^{10}}$ $+\frac{1932173}{9072 (1-T)^{11}}+\frac{496957607}{45360 (1-T)^{12}}$ $-\frac{44659349}{1440 (1-T)^{13}}+\frac{61106719}{1296 (1-T)^{14}}-\frac{48140453}{1080 (1-T)^{15}}$ $+\frac{33882455}{1296 (1-T)^{16}}-\frac{419015}{48 (1-T)^{17}}+\frac{34300}{27 (1-T)^{18}} $\\ \hline			
				$\mH_{g}(3,\varnothing;z)$&$\frac{9}{640}-\frac{3}{160 (1-T)}-\frac{1}{10 (1-T)^2}$ $+\frac{97}{320 (1-T)^3}-\frac{283}{640 (1-T)^4}$ $+\frac{79}{160 (1-T)^5}-\frac{23}{64 (1-T)^6}$ $+\frac{7}{64 (1-T)^7} $&$ -\frac{19}{6720 (1-T)^3}+\frac{287}{3840 (1-T)^4}-\frac{45707}{80640 (1-T)^5}$ $+\frac{16829}{8064 (1-T)^6}-\frac{20941}{4480 (1-T)^7}+\frac{10397}{1440 (1-T)^8}$ $-\frac{92203}{11520 (1-T)^9}+\frac{145}{24 (1-T)^{10}}-\frac{3119}{1152 (1-T)^{11}}$ $+\frac{1225}{2304 (1-T)^{12}} $\\ \hline
				$\mH_{g}(32,\varnothing;z)$&$-\frac{39}{80}+\frac{3}{5 (1-T)}+\frac{77}{32 (1-T)^2}$ $-\frac{185}{32 (1-T)^3}+\frac{141}{32 (1-T)^4}$ $-\frac{17}{160 (1-T)^5}-\frac{781}{160 (1-T)^6}$ $+\frac{261}{32 (1-T)^7}-\frac{187}{32 (1-T)^8}$ $+\frac{49}{32 (1-T)^9} $&$ \frac{23}{1120 (1-T)^2}+\frac{97}{320 (1-T)^3}-\frac{14673}{4480 (1-T)^4}$ $+\frac{117595}{8064 (1-T)^5}-\frac{1413481}{40320 (1-T)^6}+\frac{1670927}{40320 (1-T)^7}$ $+\frac{142223}{40320 (1-T)^8}-\frac{596659}{5760 (1-T)^9}+\frac{240427}{1152 (1-T)^{10}}$ $-\frac{1401067}{5760 (1-T)^{11}}+\frac{101459}{576 (1-T)^{12}}-\frac{41659}{576 (1-T)^{13}}$ $+\frac{1225}{96 (1-T)^{14}} $\\ \hline
				$\mH_{g}(3,2^2;z)$&$-\frac{T}{8}-\frac{11}{40}+\frac{3}{8 (1-T)}$ $-\frac{271}{480 (1-T)^2}-\frac{113}{96 (1-T)^3}$ $+\frac{1133}{120 (1-T)^4}-\frac{223}{15 (1-T)^5}$ $+\frac{131}{240 (1-T)^6}+\frac{9629}{240 (1-T)^7}$ $-\frac{3689}{40 (1-T)^8}+\frac{843}{8 (1-T)^9}$ $-\frac{1937}{32 (1-T)^{10}}+\frac{441}{32 (1-T)^{11}} $&$-\frac{1}{84 (1-T)^2}+\frac{409}{6720 (1-T)^3}-\frac{22447}{13440 (1-T)^4}$ $+\frac{238163}{17280 (1-T)^5}-\frac{3407293}{60480 (1-T)^6}+\frac{1313971}{10080 (1-T)^7}$ $-\frac{922919}{8640 (1-T)^8}-\frac{11308307}{30240 (1-T)^9}+\frac{1884201}{1120 (1-T)^{10}}$ $-\frac{1568119}{432 (1-T)^{11}}+\frac{88766551}{17280 (1-T)^{12}}-\frac{28419997}{5760 (1-T)^{13}}$ $+\frac{1773775}{576 (1-T)^{14}}-\frac{644467}{576 (1-T)^{15}}+\frac{8575}{48 (1-T)^{16}}$\\ \hline
			\end{longtable}
			\tiny{\textbf{Table A}: $\mH_g(\mu,\nu;z)$}
		\end{center}
		\begin{center}\label{Cd}
			\renewcommand\arraystretch{2}
			\tabcolsep=0.1cm
			\begin{longtable}{|p{2.25cm}|p{3.2cm}|p{3.4cm}|p{3.45cm}|} 
				\hline
				\tiny
				\diagbox{\,\quad$\mu^{(1)}$\quad\,}{$\mu^{(2)}$\quad}&\centering{$1^d$} &\centering{$2\,1^{d-2}$} & \,\,\quad\qquad$2^2\,1^{d-4}$\\ \hline
				$C_{2}(2,{\mu^{(2)}} ;x)$&$\frac{\sinh (x)}{2} $&$ \frac{\cosh (x)}{2} $&$ ~$\\  \hline
				$C_{4}(2\,1^2,{\mu^{(2)}} ;x)$&$\frac{\sinh (x)}{2}-\frac{\sinh (2 x)}{16} $ $-\frac{\sinh (3 x)}{6} +\frac{\sinh (6 x)}{48}  $&$ \frac{\cosh (x)}{2}-\frac{\cosh (2 x)}{8} $ $-\frac{\cosh (3 x)}{2} +\frac{\cosh (6 x)}{8}  $&$ -\frac{3\sinh (2 x)}{16}+\frac{\sinh (6 x)}{16}   $\\  \hline
				$C_{5}(2^2\,1,{\mu^{(2)}} ;x)$&$-\frac{19}{120}+\frac{5\cosh (2 x)}{24}$ $ -\frac{\cosh (4 x)}{24} -\frac{\cosh (6 x)}{96} $ $+\frac{\cosh (10 x)}{480}  $&$ \frac{5\sinh (2 x)}{12} -\frac{\sinh (4 x)}{6}$ $ -\frac{\sinh (6 x)}{16} +\frac{\sinh (10 x)}{48}  $&$ \frac{1}{8} -\frac{\cosh (4 x)}{8} -\frac{\cosh (6 x)}{32} $ $+\frac{\cosh (10 x)}{32}  $\\  \hline
				$C_{7}(2^3\,1,{\mu^{(2)}} ;x)$&$\frac{97 \sinh (x)}{640}-\frac{61\sinh (3 x)}{864} $ $+\frac{\sinh (5 x)}{160} +\frac{13 \sinh (7 x)}{2688}$ $+\frac{\sinh (9 x)}{5760}-\frac{\sinh (11 x)}{1920}$ $-\frac{\sinh (15 x)}{17280}+\frac{\sinh (21 x)}{120960} $&$ \frac{97 \cosh (x)}{640}-\frac{61\cosh (3 x)}{288}$ $ +\frac{\cosh (5 x)}{32} +\frac{13\cosh (7 x)}{384} $ $+\frac{\cosh (9 x)}{640} -\frac{11 \cosh (11 x)}{1920}$ $-\frac{\cosh (15 x)}{1152}+\frac{\cosh (21 x)}{5760} $&$ -\frac{23\sinh (x)}{128}  -\frac{37\sinh (3 x)}{288} $ $+\frac{3\sinh (5 x)}{32} +\frac{19\sinh (7 x)}{384} $ $-\frac{11 \sinh (9 x)}{1152} -\frac{5\sinh (11 x)}{384} $ $-\frac{\sinh (15 x)}{384} +\frac{\sinh (21 x)}{1152} $\\  \hline
				$C_{3}(3,{\mu^{(2)}} ;x)$&$-\frac{1}{9}+\frac{\cosh (3 x)}{9}  $&$ \frac{\sinh (3 x)}{3}  $&$ - $\\  \hline
				$C_{5}(3\,1^2,{\mu^{(2)}} ;x)$&$-\frac{1}{12}+\frac{\cosh (x)}{18}$ $-\frac{\cosh (2 x)}{24} +\frac{\cosh (3 x)}{9} $ $-\frac{\cosh (4 x)}{36} +\frac{\cosh (5 x)}{90}$ $ -\frac{\cosh (6 x)}{36} +\frac{\cosh (10 x)}{360}  $&$ \frac{\sinh (x)}{18}-\frac{\sinh (2 x)}{12} $ $+\frac{\sinh (3 x)}{3} -\frac{\sinh (4 x)}{9} $ $+\frac{\sinh (5 x)}{18} -\frac{\sinh (6 x)}{6} $ $+\frac{\sinh (10 x)}{36}  $&$ \frac{1}{12}+\frac{\cosh (2 x)}{24} -\frac{\cosh (4 x)}{12} $ $-\frac{\cosh (6 x)}{12} +\frac{\cosh (10 x)}{24} $\\  \hline
				$C_{7}(3^2\,1,{\mu^{(2)}} ;x)$&$\frac{52}{2835}-\frac{35 \cosh (3 x)}{1296}$ $+\frac{4\cosh (6 x)}{405}-\frac{7 \cosh (9 x)}{6480}$ $-\frac{\cosh (15 x)}{6480} +\frac{\cosh (21 x)}{45360} $&$ -\frac{35\sinh (3 x)}{432}  +\frac{8\sinh (6 x)}{135} $ $-\frac{7\sinh (9 x)}{720} -\frac{\sinh (15 x)}{432} $ $+\frac{\sinh (21 x)}{2160} $&$ -\frac{2}{27}+\frac{5\cosh (3 x)}{144} $ $+\frac{2\cosh (6 x)}{27} -\frac{13\cosh (9 x)}{432} $ $-\frac{\cosh (15 x)}{144} +\frac{\cosh (21 x)}{432}  $\\  \hline
				$C_{7}(3\,2^2,{\mu^{(2)}} ;x)$&$\frac{25}{1512}+\frac{\cosh (x)}{80}$ $-\frac{5\cosh (2 x)}{288} -\frac{\cosh (3 x)}{48} $ $-\frac{\cosh (4 x)}{360} +\frac{5\cosh (5 x)}{576} $ $+\frac{\cosh (6 x)}{270} -\frac{\cosh (7 x)}{4032} $ $+\frac{\cosh (9 x)}{1728}-\frac{\cosh (11 x)}{1440}$ $-\frac{\cosh (14 x)}{10080}+\frac{\cosh (21 x)}{60480} $&$ \frac{\sinh (x)}{80}-\frac{5\sinh (2 x)}{144} $ $-\frac{\sinh (3 x)}{16} -\frac{\sinh (4 x)}{90} $ $+\frac{25\sinh (5 x)}{576} +\frac{\sinh (6 x)}{45} $ $-\frac{\sinh (7 x)}{576} +\frac{\sinh (9 x)}{192} $ $-\frac{11 \sinh (11 x)}{1440} -\frac{\sinh (14 x)}{720} $ $+\frac{\sinh (21 x)}{2880} $&$-\frac{1}{72} -\frac{\cosh (x)}{24}$ $-\frac{\cosh (2 x)}{96} +\frac{\cosh (4 x)}{72} $ $+\frac{11\cosh (5 x)}{192} +\frac{\cosh (6 x)}{72} $ $+\frac{\cosh (7 x)}{576} -\frac{\cosh (9 x)}{576} $ $-\frac{5\cosh (11 x)}{288} -\frac{\cosh (14 x)}{288} $ $+\frac{\cosh (21 x)}{576}  $\\  \hline
			\end{longtable}
			\tiny{\textbf{Table B}: $C_{d}(\mu^{(1)},{\mu^{(2)}};x)$}
		\end{center}

	\end{CJK}
\end{document}